\documentclass[11pt]{article}
\usepackage{odonnell}

\newcommand{\rnote}[1]{\footnote{\color{blue}Ryan: {#1}}}

\newcommand{\pnote}[1]{\footnote{\color{ForestGreen}Pedro: {#1}}}

\newcommand{\excess}{\mathrm{exc}}
\newcommand{\symm}[1]{S_{#1}}

\newcommand{\randB}{\bB}
\newcommand{\Cyc}{\mathrm{Cyc}}

\begin{document}

\title{Explicit near-Ramanujan graphs of every degree}
\author{Sidhanth Mohanty\thanks{EECS Department, University of California Berkeley.  Supported by NSF grant CCF-1718695} \and Ryan O'Donnell\thanks{Computer Science Department, Carnegie Mellon University.  Supported by NSF grant CCF-1717606. This material is based upon work supported by the National Science Foundation under grant numbers listed above. Any opinions, findings and conclusions or recommendations expressed in this material are those of the author and do not necessarily reflect the views of the National Science Foundation (NSF).}
 \and Pedro Paredes${}^{\text{\textdagger}}$}
\date{\today}
\maketitle

\begin{abstract}
    For every constant $d \geq 3$ and $\eps > 0$, we give a deterministic $\poly(n)$-time algorithm that outputs a $d$-regular graph on $\Theta(n)$ vertices that is $\eps$-near-Ramanujan; i.e., its eigenvalues are bounded in magnitude by $2\sqrt{d-1} + \eps$ (excluding the single trivial eigenvalue of~$d$).
\end{abstract}

\thispagestyle{empty}
\newpage
\setcounter{page}{1}

\section{Introduction}  \label{sec:intro}
In this work, we obtain explicit $d$-regular $\eps$-near-Ramanujan graphs for every~$d \geq 3$ and every $\eps > 0$.  As an example, we give the first explicit family of $7$-regular graphs with $\lambda_2(G), |\lambda_n(G)| \leq 2\sqrt{6} + \eps$. Our main result is the following:
\begin{theorem}                                     \label{thm:us}
    For any $d \geq 3$ and any $\eps > 0$, there is an explicit (deterministic polynomial-time computable) infinite family of $d$-regular graphs $G$ with $\max\{\lambda_2(G), |\lambda_n(G)|\} \leq 2\sqrt{d-1} + \eps$.
\end{theorem}
The key technical result that we prove in service of this is the following:
\begin{theorem}                                     \label{thm:key}
    Let $G$ be an arbitrary $d$-regular $n$-vertex graph.  Assume that the $r$-neighborhood of every vertex contains at most one cycle, where $r \gg (\log \log n)^2$.  Then a random edge-signing of $G$ has all its eigenvalues bounded in magnitude by $2\sqrt{d-1} + o_n(1)$, with high probability.
\end{theorem}
See \Cref{sec:ideas} for a comparison of \Cref{thm:key} with a similar theorem of Bilu and Linial~\cite{BL06}, which has an alternate hypothesis and a weaker conclusion.

\subsection{On near-Ramanujan graphs}
Let us put our results into context. Loosely speaking, \emph{expander graphs} are sparse graphs in which every small set of vertices has many edges on its boundary.  For an early paper working out relationships between various possible definitions, see Alon~\cite{Alo86}. For a thorough reference describing expanders' myriad applications and connections to various parts of computer science and mathematics, see the survey of Hoory, Linial, and Wigderson~\cite{HLW06}.

A good way to quantify the definition of expansion is through the eigenvalues of the graph.
\begin{definition}[Graph eigenvalues]
    Let $G$ be an $n$-vertex $d$-regular multigraph. We write \mbox{$\lambda_i = \lambda_i(G)$} for the \emph{eigenvalues} of its adjacency matrix~$A$, and we always assume they are ordered with $\lambda_1 \geq \lambda_2 \geq \cdots \geq \lambda_n$.  A basic fact is that $\lambda_1 = d$ always; this is called the \emph{trivial} eigenvalue.  We also  write $\lambda = \lambda(G) = \max\{\lambda_2, |\lambda_n|\}$.
\end{definition}
The extent to which a $d$-regular graph~$G$ is ``expanding'' is governed by the magnitude of its nontrivial eigenvalues; in particular, by $\lambda_2$ and (to a lesser extent)~$|\lambda_n|$.  Together these are captured by the parameter~$\lambda(G)$.  The smaller $\lambda(G)/d$ is, the better $G$'s expansion; typically, a graph is called expanding when this ratio is bounded away from~$1$.
\begin{definition}[Spectral expanders]
    An infinite sequence of $d$-regular (multi)graphs $(G_n)$ is said to be a \emph{family of expanders} if there is a constant $\delta > 0$ such that $\lambda(G_n) \leq (1-\delta)d$ for all~$n$.
\end{definition}
Pinsker~\cite{Pin73} introduced the English terminology and showed that random (bipartite) graphs have positive expansion properties with high probability (see also~\cite{BK67,LP73}).  Indeed, it can be shown~\cite{Alo86} that a uniformly random $d$-regular graph  has $\lambda_2(G) < (1-\delta)d$ with high probability, for some universal $\delta > 0$ (see \Cref{thm:friedman} below for a much stronger result).  However, for almost all of the numerous practical applications of expanders in theoretical computer science (error correcting codes, derandomization, complexity theory, cryptography, metric embeddings, etc.) it is important for the graphs to be \emph{explicit} --- i.e., constructible by a deterministic polynomial-time algorithm.  Indeed, it is even better if they are \emph{strongly explicit}, meaning that their adjacency list is computable in $\polylog n$ time.  

\subsection{Review of Ramanujan and near-Ramanujan families}
Margulis~\cite{Mar73} was the first to provide an explicit expander family; a slight variant of it, which is $8$-regular, was shown~\cite{GG81} to have $\lambda \leq 5\sqrt{2} \approx 7.1$ (see~\cite{HLW06}).  A natural question then is to provide explicit $d$-regular expanders, for various values of~$d$, with~$\lambda$ as small as possible as a function of~$d$.  The well-known \emph{Alon--Boppana bound} shows that $2\sqrt{d-1}$ is essentially a lower bound:
\begin{theorem}                                     \label{thm:alon-boppana}
    (\cite{Alo86,Nil91,Fri93}.) Let $G$ be an $n$-vertex $d$-regular multigraph. Then $\lambda_2(G) \geq 2\sqrt{d-1} - O(1/\log^2 n)$.
\end{theorem}
On the other hand, using the resolution of the Ramanujan--Petersson conjectures in various number-theoretic settings, 
it is possible to construct $d$-regular expander families that meet the bound $\lambda(G) \leq 2\sqrt{d-1}$ for some values of~$d$.  Lubotzky--Phillips--Sarnak~\cite{LPS88} dubbed such graphs \emph{Ramanujan}.
\begin{definition}[Ramanujan graphs]
    A $d$-regular (multi)graph~$G$ is called \emph{(two-sided) Ramanujan} whenever $\lambda(G) \leq 2\sqrt{d-1}$.  When we merely have $\lambda_2(G) \leq 2\sqrt{d-1}$, we call $G$ \emph{one-sided Ramanujan}; if $G$ is bipartite this implies that $|\lambda_i| \leq 2\sqrt{d-1}$ for all $i \neq 1,n$, with $\lambda_n(G) = -d$ being inevitable.
\end{definition}
\noindent We remark that some expander properties (e.g., edge-expansion for small sets) only need a one-sided eigenvalue bound, whereas others (e.g., the Expander Mixing Lemma) need a two-sided bound.

Regarding the explicit construction of $d$-regular Ramanujan graphs using number theory, the case when $d-1$ is an odd prime is due to Ihara~\cite{Iha66} (implicitly) and to Lubotzky--Phillips--Sarnak~\cite{LPS88} and Margulis~\cite{Mar88} (independently); the  $d-1 = 2$ case is by Chiu~\cite{Chi92}; and, the general prime power case mentioned below is due to Morgenstern~\cite{Mor94}.  For extensions to general~$d$ where the eigenvalue bound depends on the number of distinct prime divisors of~$d-1$, see~\cite{Piz90,Cla06}.
\begin{theorem}
    (\cite{Mor94}.) For any $d \geq 3$ with $d-1$ a prime power, there is a strongly explicit family of $d$-regular Ramanujan graphs.
\end{theorem}
For all other values of~$d$ --- e.g., for $d = 7$ --- it is unknown if infinite families of $d$-regular Ramanujan graphs exist (but see \Cref{thm:mss} below for the one-sided bipartite case).  However, it is known that \emph{near}-Ramanujan graph families exist for every~$d$. 
Alon~\cite{Alo86} conjectured that a random $n$-vertex $d$-regular graph~$\bG$ has $\lambda(\bG) \leq 2\sqrt{d-1} + o_n(1)$ with high probability, and this was proven two decades later by Friedman~\cite{Fri08}.  Bordenave~\cite{Bor19} has recently given a simpler proof, and our paper will involve modifying and derandomizing Bordenave's work.
\begin{theorem}                                     \label{thm:friedman}
    (\cite{Fri08}.)  Fix any $d \geq 3$ and $\eps > 0$ and let $\bG$ be a uniformly random $d$-regular graph.  Then
    \[
        \Pr\bracks*{\lambda(\bG) \leq 2\sqrt{d-1} + \eps} \geq 1 - o_n(1).
    \]
    In fact~\cite{Bor19}, $\bG$ achieves the  subconstant $\eps = \wt{O}(1/\log^2 n)$ with probability at least \mbox{$1-1/n^{.99}$}.
\end{theorem}
A natural question then is whether, for every~$d$, one can achieve \emph{explicit} graph families that are ``$\eps$-near-Ramanujan'' as above.  In their work introducing the \emph{zig-zag product}, Reingold--Vadhan--Wigderson~\cite{RVW02} asked whether explicit families could at least reach a bound of~$O(\sqrt{d})$; towards this, their work gave strongly explicit families with $\lambda(G) \leq O(d^{2/3})$.  By extending their approach, Ben-Aroya and Ta-Shma reached~$d^{1/2 + o(1)}$:
\begin{theorem}                                     \label{thm:zigzag}
    (\cite{RVW02,BT11}.)  There are strongly explicit families of $d$-regular multigraphs~$G$ satisfying the bound $\lambda(G) \leq \sqrt{d} \cdot 2^{O(\sqrt{\log d})}$.
\end{theorem}
Bilu and Linial~\cite{BL06} got even closer to~$O(\sqrt{d})$, using a new approach based on random \emph{lifts} that will prove important in our paper.  Their graph families are not strongly explicit, although Bilu--Linial point out they are at least ``probabilistically strongly explicit'' (q.v.~\Cref{thm:prob-strong-explicit}).
\begin{theorem}                                     \label{thm:BL}
    (\cite{BL06}.)  There are explicit families of $d$-regular multigraphs~$G$ satisfying the bound $\lambda(G) \leq \sqrt{d} \cdot O(\log^{1.5} d)$.
\end{theorem}
Due to their asymptotic-in-$d$ nature, neither of \Cref{thm:zigzag,thm:BL} gives much help for specific small values of~$d$ not covered by Morgenstern, such as $d = 7$.  In such cases, one can use a simple idea due to Cioab\u{a} and Murty~\cite{CM08} (cf.~\cite{HM06}): take a prime (or prime power) \mbox{$q <  d-1$}, form a $(q+1)$-regular Ramanujan graph, and then add in $d-q-1$ arbitrary perfect matchings.  It is shown in~\cite{CM08} that each perfect matching increases~$\lambda(G)$ by at most~$1$.  Hence:
\begin{theorem}                                     \label{thm:CM}
    (\cite{CM08}.) For any $d \geq 3$, there is a strongly explicit family of $d$-regular multigraphs with $\lambda(G) \leq 2\sqrt{d-1} + \text{gap}(d)$, where $\text{gap}(d)$ denotes the least value~$g$ such that $d-1-g$ is a prime (power).  One can bound $\text{gap}(d)$ by $O(\log^2 d)$ under Cram\'{e}r's conjecture, by $O(\sqrt{d} \log d)$ under the Riemann Hypothesis, or by $O(d^{.525})$ unconditionally.
\end{theorem}
\noindent For example, this gives strongly explicit $7$-regular multigraphs with $\lambda(G) \leq 2\sqrt{5} + 1 < 5.5$.  For comparison, the Ramanujan bound is $2\sqrt{6} < 4.9$.

Finally, Marcus--Spielman--Srivastava~\cite{MSS15a,MSS15d} recently introduced the \emph{Interlacing Polynomials Method} and used it to show that \emph{one-sided bipartite} Ramanujan graphs exist for all~\mbox{$d \geq 3$} and all even~$n$.  Their proof was merely existential, but Cohen~\cite{Coh16} was able to make it explicit (though not strongly so):
\begin{theorem}                                     \label{thm:mss}
    (\cite{MSS15a,MSS15d,Coh16}.) For any $d \geq 3$, there is an explicit family of one-sided bipartite, $d$-regular, Ramanujan multigraphs.
\end{theorem}
As mentioned, this theorem gives an $n$-vertex graph for every even~$n$, which is slightly better than all other results mentioned in this section, which merely give graphs for a dense sequence of~$n$'s (typically, a sequence $n_j$ with $n_{j+1} - n_j = o(n_j)$).  Also, as pointed out to us by Nikhil Srivastava, pairing left and right vertices in the construction from \Cref{thm:mss} and merging them gives ``twice-Ramanujan'' graphs of every even degree; i.e., $2d$-regular graphs for all $d\ge 3$ with $\lambda(G)\le4\sqrt{d-1}$.\footnote{We include a short proof here: let $\wt{A} = \begin{bmatrix}0 & A\\ A^T & 0 \end{bmatrix}$ be the adjacency matrix of a $d$-regular bipartite Ramanujan graph.  Then $A+A^T$ is the adjacency matrix of the merged graph.  For any $x$ orthogonal to $\vec{1}$,  $(A+A^T)x = \begin{bmatrix}\Id&\Id\end{bmatrix}\wt{A}\begin{bmatrix}x\\ x\end{bmatrix}$.  Thus $\left\|(A+A^T)x\right\| \le \sqrt{2}\cdot\left\|\wt{A}\begin{bmatrix}x\\ x\end{bmatrix}\right\|$.  Since $\begin{bmatrix}x\\ x\end{bmatrix}$ is orthogonal to both $\begin{bmatrix}\vec{1}\\ \vec{1}\end{bmatrix}$ and $\begin{bmatrix}\vec{1}\\ -\vec{1}\end{bmatrix}$, we have $\left\|{\wt{A}\begin{bmatrix}x \\ x\end{bmatrix}}\right\| \le 2\sqrt{d-1}\sqrt{2}\norm{x}$.  One can then conclude that $\|(A+A^T)x\|\le 4\sqrt{d-1}\|x\|$.}  One can then obtain $(2d+1)$-regular graphs with $\lambda(G)\le 4\sqrt{d-1}+1$ by adding an arbitrary perfect matching via the result of \cite{CM08}.

\newcommand{\specialcell}[2][c]{%
  \begin{tabular}[#1]{@{}c@{}}#2\end{tabular}}
\newcommand{\dingcheck}{{\small \ding{51}}}
\newcommand{\dingboldcheck}{{\small \ding{52}}}
\newcommand{\dingx}{{\small \ding{55}}}
\newcommand{\dingboldx}{{\small \ding{56}}}
\begin{table}[ht] \label{tab:previous-constructions}
\caption{Comparison of our \Cref{thm:us} with prior work.}
\begin{tabular}{lcccccc}
\hline
 {\tiny \textbf{Who?}} &
 {\tiny \textbf{Which \emph{d}?}} &
 \specialcell{{\tiny \textbf{Eigenvalue}} \\ {\tiny \textbf{bound}}}  &
 {\tiny \textbf{2-sided?}} &
 \specialcell{{\tiny \textbf{Strongly}} \\ {\tiny \textbf{explicit?}}} &
 \specialcell{{\tiny \textbf{Always}} \\ {\tiny \textbf{simple?}}} &
\specialcell{{\tiny \textbf{\# vertices}} \\ {\tiny \textbf{given \emph{n}}}}
 \\ \hhline{=======}
 {\tiny \cite{Iha66,LPS88,Mar88,Chi92,Mor94}} & {\tiny prime power\,+\,1} & {\tiny $2\sqrt{d-1}^{\vphantom{i}}$} & \dingcheck & \dingcheck  & \dingcheck & {\tiny $n(1+ o(1))$} \\
 {\tiny \cite{Piz90,Cla06}}& {\tiny any $d$}  & {\tiny $2^{\text{om}(d-1)} \sqrt{d-1}$} \ {\tiny ${}^\ast$} & \dingcheck & \dingcheck
 & \dingx
 & {\tiny $n(1 + o(1))$ } \\
 {\tiny \cite{RVW02,BT11}}& {\tiny any $d$} & {\tiny $\sqrt{d} \cdot 2^{O(\sqrt{\log d})}$} & \dingcheck & \dingcheck & \dingx & {\tiny $\Theta(n)$}
 \\
{\tiny \cite{HM06,CM08}} & {\tiny any $d$} &  {\tiny $\begin{cases}  2\sqrt{d-1} + O(\log^2 d) & {}^\dagger \\ \sqrt{d} \cdot O(\log d) & {}^\ddagger \\ O(d^{.525}) &  \end{cases}$} & \dingcheck & \dingcheck & \phantom{{\tiny ${}^\S$}}\dingcheck {\tiny ${}^\S$}  & {\tiny $n(1 + o(1))$ }  \\
 {\tiny \cite{BL06}} & {\tiny any $d$}  & {\tiny $\sqrt{d} \cdot O(\log^{1.5} d)$} & \dingcheck & \phantom{{\tiny ${}^\P$}}\dingx {\tiny ${}^\P$}& \dingcheck & {\tiny $n(1 + o(1))$ }  \\
 {\tiny \cite{MSS15a,MSS15d,Coh16}} & {\tiny any $d$}  & {\tiny $2\sqrt{d-1}$} & \dingx & \dingx & \dingx & {\tiny $2\lceil n/2 \rceil$ }  \\  \hline
 {\tiny \textbf{this paper}} & {\tiny any $d$}  & {\tiny $2\sqrt{d-1} + \eps$} & \dingcheck & \phantom{{\tiny ${}^\P$}}\dingx {\tiny ${}^\P$} & \dingcheck & {\tiny $n(1 + o(1))$ } \\ \hline
\end{tabular}

{\tiny ${}^\ast$ In the this entry we have written $\text{om}(d-1)$ for the number of distinct prime divisors of $d-1$.  Thus~\cite{Cla06} generalizes the preceding ``prime power + 1'' entry of~\cite{Mor94}.  Also, $2^{\text{om}(d-1)}$ is at most $2^{O(\log d / \log \log d)} = d^{o(1)}$ for all~$d$, and is $(2+o(1))^{\ln \ln d} = O(\log^{0.7} d)$ for ``most''~$d$.}

{\tiny ${}^\dagger$ Assuming Cram\'{e}r's conjecture. \quad ${}^\ddagger$ Assuming the Riemann Hypothesis.}

{\tiny ${}^\S$ The construction can be made simple at the expense of making it not strongly explicit.}

{\tiny ${}^\P$ The construction is ``probabilistically strongly explicit''.}

\end{table}

\paragraph{Our results.}
As mentioned, our \Cref{thm:us} gives $\poly(n)$-time deterministically computable $n$-vertex $d$-regular graphs~$G$ with $\lambda(G) \leq 2\sqrt{d-1} + \eps$, for any $d \geq 3$ and $\eps > 0$.  To be more precise, the running time of our algorithm is $n^{f(d,\eps)}$ where $f(d,\eps) = O(d^{1/4}\log(d)/\sqrt{\eps})$.  Although our graphs are not strongly explicit, they are ``probabilistically $\poly\log n$-time computable'', a relaxation of a notion defined by~\cite{BL06}.  Essentially, this means we show there \emph{exist} near-Ramanujan graphs whose adjacency lists are computable in $\polylog n$ time, and furthermore there is a $\polylog(n)$-time randomized algorithm for finding them with high probability.  More precisely, the following statement holds:
\begin{theorem}                                     \label{thm:prob-strong-explicit}
    There is a deterministic polynomial-time algorithm with the following properties:
    \begin{itemize}
        \item It takes as input $N$, $d \geq 3$, and $\eps > 0$ written as binary strings.
        \item It also takes as input a ``seed'' $s \in \{0,1\}^{O(\log^2 N)}$ (the $O(\cdot)$ hides a factor of~$O(d^{1/4}\log(d)/\sqrt{\eps})$).
        \item It outputs a Boolean circuit~$C$ that implements the ``adjacency list'' of a $d$-regular graph~$G$ on $N' \sim N$ vertices in $\polylog(N)$ time. (This means that on input $u \in [N']$ and $i \in [d]$, both expressed in binary, $C(u,i)$ outputs the $v \in [N']$ that is the $i$th neighbor of~$u$ in~$G$.)
        \item With high probability over the choice of seed~$s$, the resulting graph $G$ satisfies the bound \mbox{$\lambda(G) \leq 2\sqrt{d-1} + \eps$}.
    \end{itemize}
\end{theorem}
The difference betweeen our notion and that of \cite{BL06} is that we use a seed of length-$O(\log^2 N)$, whereas the notion in \cite{BL06} requires the seed length to be $O(\log N)$.

\subsection{On Bordenave's theorem with random edge-signs}  \label{sec:on-bordenave}
Since our result may be viewed as a derandomization of the Friedman/Bordenave theorem (\Cref{thm:friedman}), let us take some time to describe this result.  Friedman's original proof is notably quite involved ($100$~pages).  Bordenave's proof is certainly simpler (more like~$30$ pages), although it is by no means easy.  However, Bordenave's proof can become still simpler if one is willing consider a variant: when $\bG$ is not just a random $d$-regular graph, but rather a \emph{randomly edge-signed} random $d$-regular graph.

Let us say a few words about why this makes things simpler.  First, it turns out that in this case one need not worry about the ``trivial eigenvalue'' of~$d$; it no longer exists, and the statement to be proven is simply that $\rho(\bG) \leq 2\sqrt{d-1} + \eps$ with high probability, where $\rho(\bG)$ is the spectral radius (largest eigenvalue-magnitude) of the (signed) adjacency matrix of~$\bG$.  Second, with random edge-signs, each entry of~$\bG$'s adjacency matrix becomes a symmetric random variable, and it is always more pleasant in probability theory when one's random variables naturally have mean zero.

In fact, there are scenarios in which one might actually \emph{want} to consider random edge-signed $d$-regular graphs.  For example when studying the Max-Cut problem, the setting of sparse random graphs is a very natural and challenging one, and many algorithms/complexity results depend on eigenvalue bounds for such graphs. Having random edge-signs simply means studying the equally natural $2$XOR (aka~$2$Lin) problem, one that has a long history in theoretical computer science as well~\cite{Has84}.

Undoubtedly experts would know that including random edge-signs should make Bordenave's proof simpler, but it doesn't appear to have been directly explored until the recent work of Deshpande et~al.~\cite{DMOSS19}.  That paper proved the analogue of Friedman/Bordenave for random edge-signings of random $(c,d)$-biregular graphs.  The case when $c = d$ is essentially the same as the $d$-regular random graph case, but the nature of the proof simplification is perhaps obscured, particularly because~\cite{DMOSS19} directly cited several lemmas from Bordenave~\cite{Bor19}. A similar situation occurred in a subsequent work~\cite{MOP19}, which has random edge-signs within an even more complicated random graph model.

In fact, a side motivation we had for this paper was to carefully set out a self-contained proof --- as simple as possible --- of ``Alon's Conjecture'' for randomly edge-signed graphs. A reader not interested in derandomization may nevertheless find our proof of the below theorem of interest, particularly since it contains a substantial portion of Bordenave's proof of Friedman's theorem.
\begin{theorem}                                     \label{thm:simple}
    Let $d \geq 3$ and $\eps > 0$.  If $\bG$ is a random edge-signed $d$-regular $n$-vertex graph, then
    \[
        \Pr\bracks*{\rho(\bG) \leq 2\sqrt{d-1} + \eps} \geq 1 - o_n(1).
    \]
\end{theorem}
In the course of proving this theorem, we are able to observe that in fact \Cref{thm:key} holds.  That is, \Cref{thm:simple} does not thoroughly rely on having a random edge-signing of a \emph{random} $d$-regular graph.  Instead, it works for a random edge-signing of \emph{any} $d$-regular graph that has one particular property: namely, every vertex-neighborhood of radius $O((\log \log n)^2)$ should have at most one cycle.  This property --- called tangle-freeness by Bordenave (simplifying Friedman's notion of ``tangles'') --- is a property that random $d$-regular graphs have with high probability, even for neighborhoods of the much larger radius~$\Theta(\log_{d-1} n)$.

With \Cref{thm:key} in hand, we are in a position rather like that of Bilu--Linial, who similarly showed~\cite[Cor.~3.1]{BL06} that a random edge-signing of any sufficiently good small-set expander has spectral radius at most $\sqrt{d} \cdot O(\log^{1.5} d)$ (with high probability).  As in Bilu--Linial, it is also fairly straightforward to see that \Cref{thm:key} can be derandomized effectively using almost-$k$-wise independent binary random variables.

We next describe how this derandomized result on edge-signings leads to our main \Cref{thm:us}.

\subsection{Explicit near-Ramanujan graphs via repeated 2-lifts}  \label{sec:ideas}
Let $G = (V,E)$ be an $n$-vertex $d$-regular graph, and let $\wt{G}$ be the edge-signed version of it associated to edge-signing $w : E \to \{\pm 1\}$.  As observed by Bilu and Linial~\cite{BL06}, this edge-signing is in a sense equivalent to the ``$2$-lift'' $G_2 = (V_2, E_2)$ of~$G$ defined by
\[
    V_2 = V \times \{\pm 1\}, \qquad E_2 = \Bigl\{ \{(u,\sigma), (v,\sigma\cdot w(u,v))\} : (u,v) \in E \Bigr\}.
\]
This $G_2$ is a $2n$-vertex $d$-regular graph, and the equivalence is that $G_2$'s eigenvalues are precisely the multiset-union of $G$'s eigenvalues and $\wt{G}$'s eigenvalues.  (The latter refers to the eigenvalues of~$\wt{G}$'s signed adjacency matrix, whose nonzero entries are $w(u,v)$ for each $\{u,v\} \in E$.)  In particular, if all the eigenvalues of $G$ and $\wt{G}$ have magnitude at most $2\sqrt{d-1} + \eps$ (excluding $G$'s trivial eigenvalue of~$d$), then the same is true of~$G_2$ (excluding \emph{its} trivial eigenvalue).  Thus \Cref{thm:key} can provide us with a (derandomizable) way of doubling the number of vertices in an $\eps$-near-Ramanujan graph. It is not hard to see (\Cref{prop:lift-is-bike-free}) that if $G$ is ``$r$-bicycle-free'' --- meaning that every radius-$r$ vertex neighborhood in~$G$ has at most one cycle --- then $G_2$ will also be $r$-bicycle-free.  Thus we may repeatedly double the number of vertices in an $\eps$-near-Ramanujan graph, so long as the parameter~$r$ remains $\omega((\log \log |V|)^2)$, where $|V|$ is the ``current'' number of vertices. (Unfortunately, we do not see an obvious way to get the parameter~$r$ to increase as we perform~$2$-lifts.)  This is roughly the same strategy employed in~\cite{BL06}.

As a consequence, to obtain a final $d$-regular $\eps$-near-Ramanujan graph with~$\Theta(N)$ vertices, all we need to get started is some $d$-regular $\eps$-near-Ramanujan graph~$H$ on a smaller number of vertices,~$n$, which is $O((\log \log N)^2)$-bicycle-free.  Thanks to Friedman/Bordenave, we know that a \emph{random} \mbox{$d$-regular} $n$-vertex graph is (with high probability) near-Ramanujan, and it's not hard to show it's $\Theta(\log n)$-bicycle-free.  Thus we could get started with $H$ being a random $d$-regular graph on, say, $n = 2^{\sqrt{\log N}}$ vertices, or even something smaller like $n = \operatorname{quasipoly}(\log \log N)$.

Of course, to get a construction which is overall explicit, we need to derandomize the Friedman/Bordenave analysis for this base graph~$H$.  The advantage is we now have $\poly(N)$ time to spend on constructing a graph with $n \ll N$ vertices. A trivial exponential-time derandomization won't work, but nor do we need a polynomial-time derandomization; a quasipolynomial-time derandomization is more than sufficient.  And as we will see in \Cref{sec:bordenave-derand}, it is possible to derandomize Bordenave's proof in deterministic $n^{O(\log n)}$ time using $O(\log n)$-wise uniform permutations.  The proof of this is not completely straightforward because Bordenave's proof uses a twist on the Trace Method (since the plain Trace Method provably fails).

\section{Preliminaries}     \label{sec:prelims}

\subsection{Standard derandomization tools}
Throughout we use \textbf{boldface} to denote random variables.

\begin{definition}[$(\delta,k)$-wise uniform bits]
    Let $\delta \in [0,1]$ and $k \in \N^+$.  A sequence of Boolean random variables $\by = (\by_1, \dots, \by_n) \in \{\pm 1\}^n$ is said to be \emph{$(\delta,k)$-wise uniform}\footnote{Frequently called $(\delta,k)$-wise \emph{independent} in the literature.} if, for every $S \subseteq [n]$ with $0 < |S| \leq k$, it holds that $\abs{\E[\prod_{i \in S} \by_i]} \leq \delta$.  When $\delta = 0$, we simply say that the sequence is \emph{(truly) $k$-wise uniform}; indeed, in this case the bits are individually uniformly distributed and are $k$-wise independent.
\end{definition}

A classic result of Naor and Naor~\cite{NN93} shows that $(\delta,k)$-wise uniform bits can be constructed efficiently and deterministically from a truly random seed of length $O(\log k  + \log \log n + \log(1/\delta))$.  Indeed, these bits can be generated ``strongly explicitly'' (using~\cite{Sho90}; cf.~\cite{AGHP92}):
\begin{theorem}                                  \label{thm:nn93}
    (\cite{NN93}.)  There is a deterministic algorithm that, given $\delta$,  $k$, and $N$, runs in time $\poly(N/\delta)$ and outputs a multiset $Y \subseteq \{\pm 1\}^N$ of cardinality $S = \poly(k \log(N)/\delta)$ (a power of~$2$) such that, for $\by \sim Y$ chosen uniformly at random, the sequence $\by$ is $(\delta,k)$-wise uniform.  Indeed, if the algorithm is additionally given $1 \leq s \leq S$ and $1 \leq i \leq N$ (written in binary), it can output the $i$th bit of the $s$th string in~$Y$ in deterministic time $\polylog(N/\delta)$.
\end{theorem}
We will make use of the fact that the parameters in this theorem have excellent dependence on~$N$ and~$k$. We now discuss the analogous concept for random permutations, where it is not known if the parameter dependence can be as strong.

\begin{definition}[$(\delta,k)$-wise uniform permutations]
    Let $\delta \in [0,1]$ and $k \in \N^+$.  Let $[n]_k$ denote the set of all sequences of~$k$ distinct indices from~$[n]$.  A random permutation $\bpi \in \symm{n}$ is said to be \emph{$(\delta,k)$-wise uniform} if, for every sequence $(i_1, \dots, i_k) \in [n]_k$, the distribution of $(\bpi(i_1), \dots, \bpi(i_k))$ is $\delta$-close in total variation distance from the uniform distribution on~$[n]_k$. When $\delta = 0$, we simply say that the permutation is \emph{(truly) $k$-wise uniform}.
\end{definition}
Kassabov~\cite{Kas07} and Kaplan--Naor--Reingold~\cite{KNR09} independently obtained a deterministic construction of $(\delta,k)$-wise uniform permutations with seed length $O(k \log n + \log(1/\delta))$.  Again, the construction is even ``strongly explicit'':
\begin{theorem}             \label{thm:knr}                        (\cite{KNR09,Kas07}.)  There is a deterministic algorithm that, given $\delta$,  $k$, and $n$, runs in time $\poly(n^k/\delta)$ and outputs a multiset $\Pi \subseteq \symm{n}$ (closed under inverses) of cardinality $S = \poly(n^k/\delta)$ (a power of~$2$) such that, for $\bpi \sim \Pi$ chosen uniformly at random, $\bpi$ is a $(\delta,k)$-wise uniform permutation. Indeed, if the algorithm is additionally given $1 \leq s \leq S$ and $1 \leq i \leq n$ (written in binary), it can output  $\pi_s(i)$ and $\pi_s^{-1}(i)$ (where $\pi_s$ is the $s$th permutation in~$\Pi$) in deterministic time $\poly(k \log(n/\delta))$.
\end{theorem}
We will also use a convenient theorem of Alon and Lovett~\cite{AL13}:
\begin{theorem}                                     \label{thm:al13}
    (\cite{AL13}.) Let $\bpi \in \symm{n}$ be a $(\delta,k)$-wise uniform permutation.  Then one can define a (truly) $k$-wise uniform permutation~$\bpi' \in \symm{n}$ such that the total variation distance between $\bpi$ and~$\bpi'$ is~$O(\delta n^{4k})$.
\end{theorem}
Combining the previous two results yields the following:
\begin{corollary}                                       \label{cor:prg-perm}
    (\cite{KNR09,Kas07,AL13}) There is a deterministic algorithm that, given $k$ and~$n$, runs in time $\poly(n^k)$ and outputs a multiset $\Pi \subseteq \symm{n}$ (closed under inverses) such that, when $\bpi \sim \Pi$ is chosen uniformly at random, $\bpi$ is $n^{-100k}$-close in total variation distance to a (truly) $k$-wise uniform permutation. (And the final ``indeed'' statement from \Cref{thm:knr} also holds.)
\end{corollary}

\subsection{Elementary graph theory}
\subsubsection{Random $d$-regular graphs} \label{sec:rand-d-reg}
We will be concerned with $d$-regular (multi)graphs.  We start by describing the standard way to generate random $d$-regular graphs: the \emph{configuration model}, see~\cite{BC78,Bol80,Bol01}.
\begin{definition}[Configuration model] \label{def:config-model}
    Given integers $n > d > 0$ with $nd$ even, the \emph{configuration model} produces a random $n$-vertex, $d$-regular undirected multigraph (with loops)~$\bG$.  This multigraph is induced by a uniformly random matching~$\bM$ on the set of ``half-edges'', $[n] \times [d] \cong [nd]$ (where $(v,i) \in [n] \times [d]$ is thought of as half of the $i$th edge emanating from vertex~$v$).  We identify~$\bM$ with a symmetric matrix in $\{0,1\}^{nd \times nd}$ having $1$'s precisely in the entries corresponding to matched pairs $\{(v,i), (v',i')\}$.  We may think of~$\bM$ being generated as follows: First a uniformly random permutation  $\bpi \in \symm{nd}$ is chosen; then we set $\bM_{\bpi(j),\bpi(j+1)} = \bM_{\bpi(j+1),\bpi(j)} = 1$ for each odd $j \in [nd]$.

    Given~$\bM$, the multigraph~$\bG$ is formed by ``attaching'' the matched half-edges.  More formally, the $(v,v')$-entry of~$\bG$'s adjacency matrix $\bA$ is the sum, over all $i,i' \in [d]$, of $\bM_{(v,i),(v',i')}$.  Hence
    \[
        \bA_{v,v'} = \sum_{i,i' = 1}^d \sum_{\substack{\text{odd} \\ j \in [nd]}} (1[\bpi(j) = (v,i)] \cdot 1[\bpi(j+1) = (v',i')] + 1[\bpi(j) = (v',i')] \cdot 1[\bpi(j+1) = (v,i)]).
    \]
    Note that $\bA_{v,v}$ will always be even; a self-loop is considered to contribute degree~$2$.
\end{definition}
It is well known that a graph~$\bG$ drawn from the configuration model is simple --- i.e., has no cycles of length~$1$ or~$2$ --- with probability~$\Omega_d(1)$.  As it is pleasant to work with simple graphs, we will show in  \Cref{sec:simple} that this continues to hold for \emph{pseudorandom} $d$-regular graphs, when an $O(d^2)$-wise uniform permutation is used in the configuration model.  We also record the well known fact that for $\bG$ drawn from the configuration model, when $\bG$ is conditioned on being simple, its conditional distribution is uniformly random among all $d$-regular graphs.

Although the configuration model is the most natural way to generate large random $d$-regular graphs, the fact that it does not produce simple graphs with high probability is mildly annoying.  (In particular, this causes a slight technical hitch for establishing our ``probabilistically strongly explicit'' construction.)  To sidestep this, we will also consider the \emph{random lift} model for producing random $d$-regular graphs.
\begin{definition}[Lift model]	\label{def:lift-model}
	Fix a (simple) \emph{base} graph $\ul{G} = (\ul{V},\ul{E})$ on $\ul{n}$ vertices.  Then for $n \in \N^+$, an \emph{$n$-lift of $\ul{G}$} is graph $G$ defined by a collection of permutations $\pi_{uv} \in \symm{n}$, one for each edge $(u,v) \in \ul{\vec{E}}$, under the constraint that $\pi_{uv}=\pi_{vu}^{-1}$.  The vertex set of $G$ is $\ul{V} \times [n]$, and the edges of $G$ are given by all pairs $(u,i), (v, j)$ satisfying $(u,v) \in \ul{E}$ and $\pi_{uv}(i) = j$.  When the permutations $\pi_{uv}$ are independent and uniformly random, we call the associated graph~$\bG$ a \emph{(uniformly) random $n$-lift of $\ul{G}$}.  Observe that if $\ul{G}$ is a $d$-regular graph, then $\bG$ is always a $d$-regular (simple) graph on $\ul{n} n$ vertices.
\end{definition}
%
Bordenave~\cite{Bor19} also confirmed \Cref{thm:friedman} (the Alon Conjecture) in the case that $\bG$ is a random $n$-lift of any fixed $d$-regular Ramanujan base graph~$\ul{G}$.  The simplest case is $\ul{G} = K_{d+1}$, the complete graph on $d+1$ vertices.  This gives a way to randomly construct arbitrarily large $d$-regular near-Ramanujan graphs that are always simple.  We will also derandomize this result, as it will be convenient for our ``probabilistically strongly explicit'' construction to have guaranteed simplicity.

\subsubsection{Bicycle-freeness}
It is well known that a $d$-regular random graph is likely to have at most one cycle in any neighborhood of radius $c \log_{d-1} n$, for a certain universal $c > 0$.  (This holds in either the configuration or the random lift model.)  Let us make some definitions to codify this.
\begin{definition}[Excess]
    Given a multigraph $H = (V,E)$, its \emph{excess} is $\excess(H) = |E|-|V|$.
\end{definition}
\begin{definition}[A/uni/bi-cyclic]
    A connected multigraph $H$ with $\excess(H) = -1$, $0$, $1$ (respectively) is said to be \emph{acyclic}, \emph{unicyclic}, \emph{bicyclic} (respectively).  In either of the first two cases, we call~$H$ \emph{bicycle-free} (or \emph{at most unicyclic}).
\end{definition}
\begin{definition}[Bicycle-free at radius~$r$]
    We say a multigraph is \emph{bicycle-free at radius~$r$} if the distance-$r$ neighborhood of every vertex is bicycle-free. Another way to say this is that a breadth-first search of depth~$r$, started at any vertex, encounters at most one ``back-edge''.  We remark that this notion was termed \emph{$r$-tangle-free} by Bordenave~\cite{Bor19}.
\end{definition}
\begin{proposition} \label{prop:lift-is-bike-free}
    If $G$ is bicycle-free at radius~$r$, and $G_2$ is a $2$-lift of $G$, then $G_2$ is bicycle-free at radius $r$.
\end{proposition}
\begin{proof}
    Let $(v,i)$ be any vertex in $G_2$.  Let $H$ be the neighborhood of $v$ in $G$ and let $H_2$ be the subgraph of $G_2$ induced by $V(H)\times[2]$.  Observe that the distance-$r$ neighborhood of $(v,i)$ is contained in $H_2$, and that $\excess(H_2)\le 0$ since $\excess(H)\le 0$.  If $H_2$ is disconnected it is isomorphic to a disjoint union of two copies of $H$ and thus the distance-$r$ neighborhood of $(v,i)$ is then isomorphic to $H$.  Otherwise, if $H_2$ is connected, $\excess(H_2)\le 0$ implies that it has at most one cycle.
\end{proof}
It is easy to see that any $n$-vertex, $d$-regular graph that is bicycle-free at radius~$r$ must have \mbox{$r \lesssim \log_{d-1} n$}.  On the other hand, as mentioned earlier, a random $d$-regular graph achieves this bound up to a constant factor, and we will derandomize the proof of this fact, within the $O(\log n)$-wise uniform configuration/lift model, in  \Cref{sec:bike-free}.

In a graph that is bicycle-free at radius~$r$, by definition we have $\excess(H) \leq 0$ for all subgraphs~$H$ contained in a single distance-$r$ neighborhood.  In fact, this property is enough to guarantee that $\excess(H)$ is small for \emph{any} subgraph~$H$ with at most $\exp(r)$ vertices, regardless of whether it's contained in a single distance-$r$ neighborhood:
\begin{theorem}                                     \label{thm:sparsity}
    Let $H$ be a $v$-vertex graph that is bicycle-free at radius~$r$.  Assume $r \geq 10 \ln v$.  Then $\excess(H) \leq \frac{\ln(ev)}{r}v$.
\end{theorem}
\noindent The rest of this subsection is devoted to the proof of the above theorem of elementary graph theory.

\begin{definition}[$\Cyc_g(G)$ and girth]
    Given a graph~$G$, let $\Cyc_g(G)$ denote the collection of all cycles in~$G$ of length at most~$g$.  Recall that if $\Cyc_g(G)$ is empty then $G$ is said to have \emph{girth} exceeding~$g$.
\end{definition}
The following fact is essentially immediate from the definitions:
\begin{fact}                                        \label{fact:cycletouch}
    Suppose $G$ is bicycle-free at radius~$r$.  Then the cycles in $\Cyc_{2r}(G)$ are  vertex-disjoint.
\end{fact}
Indeed, more generally:
\begin{proposition}                                     \label{prop:cyclebfs}
    Suppose $G$ is bicycle-free at radius~$r$.  For each $C \in \Cyc_{2r}(G)$, let $C^+$ denote the collection of vertices within distance $r - \textnormal{len}(C)/2$ of~$C$.  Then the sets $\{C^+ : C \in \Cyc_{2r}(G)\}$ are pairwise disjoint.
\end{proposition}
\begin{proof}
    If $u \in C_1^+ \cap C_2^+$, the distance-$r$ neighborhood of~$u$ is enough to include both $C_1$ and~$C_2$.
\end{proof}
%
Next, let us now recall the ``Moore bound for irregular graphs''.  Suppose $H$ is a graph with $v$ vertices and $\excess(H) = \eps v$; hence $H$ has average degree $2+2\epsilon$.  If we build a breadth-first search tree from some vertex, then after depth~$t$ we would ``expect'' to encounter at least $(1+2\epsilon)^t$ vertices.  If this exceeds~$v$ --- roughly, if $t \geq (\ln v)/(2\epsilon)$ --- then the breadth-first search must encounter a cycle.  Thus we have a heuristic argument that $\text{girth}(H) \lesssim (\ln v)/\epsilon$; i.e., $\epsilon \lesssim (\ln v)/\text{girth}(H)$.  Indeed, Alon--Hoory--Linial have precisely established this kind of result; we quote their theorem in a slightly simplified form:
\begin{theorem}                                     \label{thm:ahl}
    (\cite{AHL02}.) Let $H$ be a graph with $v$ vertices, $\excess(H) = \eps v$  (for $\eps \geq 0$), and girth~$g$.  Then $v \geq (1+2\eps)^{g/2 - 3/2}$.
\end{theorem}
\begin{corollary}                                       \label{cor:ahl}
    Let $H$ be a graph with $v \geq 3$ vertices and girth~$g \geq 20\ln v$. Then $\excess(H) \leq ((2\ln v)/g)v$.
\end{corollary}
We can now prove \Cref{thm:sparsity}, which replaces ``girth'' with ``bicycle-free radius'' in the above with only a small loss in parameters.
\begin{proof}[Proof of \Cref{thm:sparsity}]
    We will show the theorem assuming~$H$ is connected (the only case we'll need).  It is an exercise to extend it to the general case by considering~$H$'s connected components.

    Let $c = |\Cyc_{2r}(H)|$.  By deleting at most~$c$ edges from~$H$ we can obtain a $v$-vertex graph~$\wt{H}$ with girth at least (in fact, exceeding)~$2r$. Applying \Cref{cor:ahl} to~$\wt{H}$, we conclude that \mbox{$\excess(H) \leq \frac{\ln v}{r}v + c$}.  Thus it remains to show $c \leq v/r$.  This is trivial if $c = 0$, and if $c = 1$ then it can only fail if $r > v$ --- but then $H$ is unicyclic and hence has excess~$0$.  Assuming then that $c \geq 2$, choose paths in~$H$ to minimally connect the~$c$ cycles of $\Cyc_{2r}(H)$.  Now for each $C \in \Cyc_{2r}(H)$, if we ``charge'' to it the $r-\textnormal{len}(C)/2$ closest path-vertices, then no vertex is charged to multiple cycles, by virtue of \Cref{prop:cyclebfs}.  If we also charge the vertices of~$C$ to itself, then for each $C \in \Cyc_{2r}(H)$ we have charged a batch of
    $
        \textnormal{len}(C) + (r-\textnormal{len}(C)/2) > r
    $  vertices, and these batches are disjoint.  Thus $cr \leq v$, i.e.\ $c \leq v/r$, as required.
\end{proof}


\subsection{Non-backtracking walks and the Ihara--Bass formula} \label{sec:iharabass}
The Friedman/Bordenave theorem ultimately uses the Trace Method to analyze the eigenvalues of random $d$-regular graphs; this involves counting closed walks in them.  As observed in~\cite{Fri08,Bor19}, it is much easier to count \emph{non-backtracking} walks, and luckily the \emph{Ihara--Bass formula} gives an easy translation between eigenvalues of the adjacency matrix of a graph and the eigenvalues of its \emph{non-backtracking matrix}.
\begin{definition}[Non-backtracking matrix~\cite{Has89}]
    Let $G = (V,E)$ be a multigraph with adjacency matrix~$A$.  Let $\vec{E}$ denote the (multi)set of all directed edges formed by replacing each undirected edge in~$E$ with two opposing directed edges.  Then $G$'s \emph{non-backtracking} matrix~$B$ has rows and columns indexed by~$\vec{E}$, with
    \[
        B_{(u_1,v_1),(u_2,v_2)} =
        \begin{cases}
            1 &\text{if $v_1 = u_2$ and $v_2 \ne u_1$,}\\
            0 &\text{otherwise}.
        \end{cases}
    \]
    (Note that this matrix is not symmetric in general.)    In case $G$ is an edge-signed graph, the entry~$1$ above should be replaced by $A_{u_2,v_2}$, the sign of~$G$ on edge~$\{u_2,v_2\}$.
\end{definition}
In a number-theoretic context, Ihara~\cite{Iha66} implicitly showed a relationship between the eigenvalues of~$A$ and~$B$ when $G$ is regular.  Serre~\cite{Ser77} and several others suggested the translation to graph theory, and Bass~\cite{Bas92} (following~\cite{Has89}) explicitly established:
\begin{theorem}                                     \label{thm:ib}
    (Ihara--Bass formula.)  Let $G$ be a $d$-regular (multi)graph and write $q = d-1$.  Then
    \[
        \det(\Id - zB) = (1-z^2)^{\excess(G)}\det((1+qz^2)\Id - zA),
    \]
    where $\Id$ denotes the identity matrix (of appropriate dimension).
\end{theorem}
This theorem has been given many proofs, and it can be generalized to irregular graphs, edge-weighted graphs, and infinite graphs.  We will use the following result, which is immediate from the edge-weighted generalization~\cite{WF09} when all weights are~$\pm 1$:
\begin{theorem}                                     \label{thm:ib2}
    (\cite{WF09}.) The Ihara--Bass formula holds  as stated above for edge-signed graphs.
\end{theorem}
The utility of Ihara--Bass is that it gives a direct correspondence between the spectra of~$A$ and~$B$. To see this, consider the zeroes of the polynomials (in~$z$) on the left- and right-hand sides.  We have that $z$ is a zero of the left-hand side precisely if $z^{-1}$ is an eigenvalue of~$B$.  On the other hand, $z$~is a zero of the right-hand side precisely if $z^{-1} = \pm 1$ or if $z^{-1}$ is such that $z^{-1} + q/z^{-1}$ is an eigenvalue of~$A$.  Thus if we want to deduce, say, the eigenvalues of~$B$ from the eigenvalues of~$A$, we have the following:
\begin{proposition}                                     \label{prop:IB}
    (Consequence of Ihara--Bass.)  Let $G = (V,E)$ be a $(q+1)$-regular edge-signed graph with adjacency matrix~$A$ and non-backtracking matrix~$B$.  Let $\lambda \neq 0, \pm 1$ be a number such that $\lambda + q/\lambda$ is an eigenvalue of~$A$.  Then $\lambda$ is an eigenvalue of~$B$.
\end{proposition}
In fact, \Cref{prop:IB} is the only consequence of Ihara--Bass we will need in this paper, and for the convenience of the reader we give a self-contained proof (inspired by~\cite{AFH15}):
\begin{proof}
    Let $f : V \to \C$ be an eigenvector for~$A$ with eigenvalue $\lambda + q/\lambda$.  Define $g : \vec{E} \to \C$ by $g_{vw} = A_{vw}f_v - \lambda f_w$.  We claim that $Bg = \lambda g$.  It then follows that $\lambda$ is an eigenvalue of~$B$, given that  $g \not \equiv 0$ (a consequence of $f \not \equiv 0$: choose $\{v,w\} \in E$ with $f_v, f_w$ not both~$0$, and then $g_{vw} = 0 = g_{wv}$ is impossible because $\lambda \neq \pm 1$).    To verify the claim, for any $uv \in \vec{E}$ we have
    \[
           (B g)_{uv}
        = \sum_{\substack{w \sim v \\ w \neq u}} A_{vw}g_{vw}
        = \sum_{w \sim v} A_{vw}(A_{vw}f_v - \lambda f_w) -  A_{vu} (A_{vu} f_v - \lambda f_u)
        = -\lambda \sum_{w \sim v} A_{vw} f_w + q f_v + \lambda A_{vu} f_u.
    \]
    But
    $
         \sum_{w \sim v} A_{vw}f_w = (Af)_v = (\lambda + q/\lambda)f_v.
    $
    Thus $(B g)_{uv} = -\lambda^2 f_v + \lambda A_{vu} f_u = \lambda g_{uv}$, as needed.
\end{proof}
When $G$ is unsigned, $A$ has a ``trivial'' eigenvalue of $d = q+1$, corresponding to $\lambda = q$; this yields the ``trivial'' eigenvalue of $q = d-1$ for~$B$.  For general edge-signed~$G$, if $\lambda = \pm \sqrt{q} = \pm \sqrt{d-1}$ in \Cref{prop:IB}, then $\lambda + q/\lambda = \pm 2\sqrt{q} = \pm 2\sqrt{d-1}$.  Thus the Ramanujan eigenvalue bound of $2\sqrt{d-1}$ for~$A$ is equivalent to the bound~$\sqrt{d-1}$ for~$B$.  As for the ``$+\eps$'', a simple calculation (appearing in~\cite{Bor19}) shows:
\begin{corollary}                                       \label{cor:IB}
    Let $G = (V,E)$ be a $d$-regular edge-signed graph ($d \geq 3$) with adjacency matrix~$A$ and non-backtracking matrix~$B$.  If $A$ has an eigenvalue of magnitude $2\sqrt{d-1} + \eps$ (for $\eps \geq 0$) then $B$ has an eigenvalue of magnitude $\sqrt{d-1} + \sqrt{\eps}\sqrt{\sqrt{q} + \eps/4} + \eps/2$ (which is $\sqrt{d-1} + \Theta(d^{1/4}\sqrt{\eps})$ for fixed~$d$ and $\eps \to 0$).
\end{corollary}

\section{On random edge-signings of fixed base graphs} \label{sec:random-signs}

In this section we will prove \Cref{thm:key}.  In fact, we will prove the following refined version:
\begin{theorem}                                       \label{thm:keyA}
    Let $G = (V,E)$ be an arbitrary $d$-regular $n$-vertex graph, where $d \leq \polylog n$.  Assume that $G$ is bicycle-free at radius~$r \gg (\log \log n)^2$.  Then
    for $\bG$ a uniformly random edge-signing of~$G$, except with probability at most~$n^{-100}$ the non-backtracking matrix~$\bB$ of~$\bG$ satisfies the spectral radius bound
    \[
        \rho(\bB) \leq \sqrt{d-1} \cdot \parens*{1+   O\parens*{\frac{(\log \log n)^2}{r}}},
    \]
    and hence (by  \Cref{cor:IB}) the signed adjacency matrix~$\bA$ of~$\bG$ satisfies the bound
    \[
        \rho(\bA) \leq 2\sqrt{d-1} \cdot \parens*{1+O\parens*{\frac{(\log \log n)^4}{r^2}}}.
    \]

     Furthermore, let $C = C(n)$ satisfy $1 \leq C \leq \polylog n$ and suppose we merely assume that the random edge-signs are $(\delta, k)$-wise uniform for $\delta \leq n^{-O(C \log d)}$ and $k \geq 2C \log n$.  Then the above bounds continue hold, with an additional additive $O(\sqrt{d}/C)$ in the $\rho(\bB)$ bound and $O(\sqrt{d}/C^2)$ in the $\rho(\bA)$ bound.
\end{theorem}

As in~\cite{Fri08,Bor19}, the proof of \Cref{thm:keyA} will use the Trace Method.  In preparation for this, we make some definitions:
\begin{definition}[Hikes] \label{def:hikes}
    Let $G = (V,E)$ be an undirected graph.  For $\ell \in \N$, we define an \mbox{\emph{$\ell$-hike}}~$\calH$ to be a closed walk in~$G$ of exactly~$2\ell$ steps which is non-backtracking except possibly between the $\ell$th and $(\ell+1)$th step.  Given an edge-signing $w : E \to \{\pm 1\}$ we write $w(\calH)$ for the product of the edge-signs that~$\calH$ traverses, counted with multiplicity. Finally, we call a hike \emph{even} (respectively, \emph{singleton-free}) if each undirected edge traversed by~$\calH$ is traversed an even number of times (respectively, at least twice).
\end{definition}
A straightforward use of the Trace Method will now imply:
\begin{proposition}                                     \label{prop:tracemethod}
    Let $\ell \in \N^+$ and define $\bT =  \tr\parens*{\bB^\ell (\randB^\top)^\ell}$ (which is an upper bound on $\rho(\bB)^{2\ell}$).  Then for a uniformly random edge-signing $\bw : E \to \{\pm 1\}$,
    \[
        \E[\bT] \leq d^2 \cdot \#\{\textnormal{even } (\ell-1)\textnormal{-hikes }  \calH \textnormal{ in } G\} \leq d^2 \cdot \#\{\textnormal{singleton-free } (\ell-1)\textnormal{-hikes }  \calH \textnormal{ in } G\}.
    \]
    Furthermore, if $\bw$ is merely $(\delta,2\ell)$-wise uniform, the bound holds up to an additive~$\delta n d^{2\ell+2}$.
\end{proposition}
\begin{proof}
    We have
    \begin{equation}\label{eqn:its-T}
        \bT = \sum_{\vec{e}_0, \vec{e}_1, \dots, \vec{e}_{2\ell-1}, \vec{e}_{2\ell} = \vec{e}_0 \in \vec{E}} \bB_{\vec{e}_0, \vec{e}_1} \bB_{\vec{e}_1, \vec{e}_2}
        \cdots \bB_{\vec{e}_{\ell-1}, \vec{e}_{\ell}}
        \bB_{\vec{e}_{\ell+1}, \vec{e}_{\ell}} \bB_{\vec{e}_{\ell+2}, \vec{e}_{\ell+1}}  \cdots \bB_{\vec{e}_{2\ell},\vec{e}_{2\ell-1}}.
    \end{equation}
    Recalling the definition of~$\bB$, one immediately sees that $\bT$ is ``something like'' the sum of $\bw(\calH)$ over all $\ell$-hikes in~$G$.  But being careful, one sees we precisely have the following:
    \begin{quotation}
        $\bT$ is equal to the sum of $\bw(\calH)$ over all ``special'' $(\ell+1)$-hikes in~$G$, where we call an $(\ell+1)$-hike \emph{special} if its $(\ell + 2)$th step is the reverse of its $(\ell+1)$th step, and the last step is the reverse of the first step.\footnote{The astute reader will note that the sign of the first/last edge in~$\calH$ is never counted in \Cref{eqn:its-T}; however it is okay to count it twice, as $\bw(\calH)$ does, since $(\pm 1)^2 = 1$.}
    \end{quotation}
    Next, we employ the following easy fact:
    \begin{fact}                                        \label{fact:easy}
        If $\bw : E \to \{\pm 1\}$ is a fully uniformly random edge-signing,  then $\E[\bw(\calH)]$ will be~$1$ if $\calH$ is an even hike, and will be~$0$ otherwise.
    \end{fact}
    Thus
    \begin{equation}    \label{eqn:evenspecial}
        \E_{\bw : E \to \{\pm 1\}}[\bT] = \#\{\text{even, special } (\ell+1)\text{-hikes }  \calH \text{ in } G\}.
    \end{equation}
    Since an $(\ell+1)$-hike involves at most $2 \ell$ undirected edges, a crude upper bound on the number of all $(\ell+1)$-hikes in~$G$ is $n d^{2\ell}$.  Thus for an edge-signing $\bw$ that is merely $(\delta,2\ell)$-wise uniform, \Cref{eqn:evenspecial} holds up to an additive $\delta n d^{2\ell}$.  Finally, every even special $(\ell+1)$-hike $\calH$ can be formed from an even $(\ell-1)$-hike~$\calH'$ by: (i)~attaching a step and its reverse to the beginning/end of~$\calH$; (ii)~attaching a step and its reverse to the midpoint of~$\calH$.  As there are at most~$(d-1)^2 \leq d^2$ choices for how to perform (i)~and~(ii), the inequality in the proposition's statement follows.
\end{proof}

At this point, edge-signs are out of the way and we are reduced to counting singleton-free hikes.  In aid of this, we borrow some terminology from~\cite{MOP19}:
\begin{definition}
    Given an $(\ell-1)$-hike $\calH$ in graph~$G$, we write $G_{\calH} = (V_{\calH}, E_{\calH})$
    for the subgraph of~$G$ formed by the union of the edges visited by~$\calH$.  We think of $G_{\calH}$ as being ``revealed'' as the $2(\ell-1)$ steps of $\calH$ are taken in order.   We classify each step of $\calH$ as either \emph{stale}, \emph{fresh}, or \emph{boundary}.  If a step of $\calH$ traverses a previously-explored edge in $G_{\calH}$ (in either direction), we call the step \emph{stale}; otherwise, if it steps to a previously-unvisited vertex, we call the step \emph{fresh}; otherwise, we call it \emph{boundary}.  For the purposes of this definition, at the beginning of $\calH$ the initial vertex is considered to be ``previously visited''.
\end{definition}
We now put bounds on the different kinds of steps.  For the fresh steps, 
we only need the singleton-free property:
\begin{proposition}                                        \label{prop:fresh-bound}
    In a singleton-free $(\ell-1)$-hike, at least half of all steps must be stale.  Thus there are fewer than $\ell$ fresh steps.
\end{proposition}
For the boundary steps of~$\calH$, it is easy to see that there are exactly $\excess(G_{\calH}) + 1$ of them.
%
%
Thus we can bound them using only the bicycle-free property.  Together with the simple bound $|V_{\calH}| \leq 2\ell$, \Cref{thm:sparsity} implies
\begin{proposition}                                     \label{prop:boundary-bound}
    If $\calH$ is an $(\ell-1)$-hike in a graph $G$ which is bicycle-free at radius~$r \geq 10 \ln(2 \ell)$, then $\calH$ has at most $O(\frac{\log \ell}{r})\cdot \ell$ boundary steps.
\end{proposition}
Finally, to handle the stale steps we group them into ``stretches''.
\begin{proposition}                                     \label{prop:stale-stretches}
    In an $(\ell-1)$-hike~$\calH$, the stale steps may be partitioned into at most $O(\frac{\log \ell}{r}) \cdot \ell$ \emph{stretches} of consecutive stale steps, each stretch having length at most~$r$, and none straddling the ``turnaround'' at step~$\ell$.
\end{proposition}
\begin{proof}
    We begin by partitioning the stale steps into maximal contiguous stretches.  It is easy to see that each of these must be preceded in~$\calH$ by a boundary step (with a single possible exception of the ``turnaround'' at step~$\ell$).  Thus \Cref{prop:boundary-bound} implies that there are at most $O(\frac{\log \ell}{r}) \cdot \ell$ maximal stretches of stale steps.  If a maximal stretch straddles the turnaround, we can split it in two.  Finally, if necessary we now subdivide the stretches into length at most~$r$.  Since there are fewer than $2\ell$ stale steps, this subdivision can be done without increasing the number of stretches by more than $2\ell/r \leq O(\frac{\log \ell}{r}) \cdot \ell$.
\end{proof}

We may now make our final estimate:
\begin{theorem}                                     \label{thm:final-estimate}
    In a $d$-regular graph~$G$ that is bicycle-free at radius $r \geq 10\ln(2\ell)$, the number of singleton-free $(\ell-1)$-hikes~$\calH$ is at most $O(\ell^3 n) \cdot (d-1)^{\ell} \cdot (dr\ell)^{O\parens{\frac{\log \ell}{r}} \cdot \ell}$.
\end{theorem}
\begin{proof}
    \newcommand{\code}{\textsc{struct}}
    Following~\cite{Bor19}, we use an encoding argument.  To each~$\calH$ we  associate a string $\code(\calH)$ over the alphabet $\{\textsc{f}, \textsc{b}, \textsc{s}\}$, where we replace each fresh step with an~$\textsc{f}$, each boundary step with a~$\textsc{b}$, and each stale stretch with an~$\textsc{s}$.  Our goal will be to show:
    \begin{claim}
        For any string $\sigma$ with $c_f, c_b, c_s$ occurrences of $\textsc{f}, \textsc{b}, \textsc{s}$ (respectively), there are no more than \mbox{$2n \cdot (d-1)^{c_f+c_b} \cdot (2r\ell)^{c_s}$} singleton-free $(\ell-1)$-hikes~$\calH$ with $\code(\calH) = \sigma$.
    \end{claim}
    \noindent Let us complete the proof of the theorem assuming this claim.  By \Cref{prop:fresh-bound,prop:boundary-bound,prop:stale-stretches}, we have the bounds
    \[
        c_f < \ell, \qquad c_b, c_s < m \coloneqq O\parens{\tfrac{\log \ell}{r}} \cdot \ell.
    \]
    Crudely, there are at most $O(\ell^3)$ possibilities for the triple $(c_f, c_b, c_s)$.  Also, the following two quantities are increasing in $c_f, c_b, c_s$:
    \[
        2n \cdot (d-1)^{c_f+c_b} \cdot (2r\ell)^{c_s}, \qquad \Sigma_{c_f, c_b, c_s} \coloneqq \#\text{ strings of } c_f \text{ } \textsc{f}\text{'s, } c_b \text{ } \textsc{b}\text{'s, } c_s \text{ } \textsc{s}\text{'s.}
    \]
    Thus we can upper-bound the number of all singleton-free $(\ell-1)$-hikes by
    \[
        O(\ell^3 n) \cdot (d-1)^{\ell + m} \cdot (2r\ell)^{m} \cdot \Sigma_{\ell, m,m} \leq O(\ell^3 n) \cdot (d-1)^{\ell} \cdot (dr\ell)^{O(m)},
    \]
    as needed, where we used the simple bound $\Sigma_{\ell,m,m} \leq \ell^{O(m)}$.

    It remains to prove the claim.  Let $\sigma$ be as given.  We may recover all possible associated~$\calH$, in a vertex-by-vertex fashion, by first specifying the  initial vertex ($n$~choices) and then proceeding through the symbols of~$\sigma$ in order.  If we are at an~$\textsc{f}$ or a~$\textsc{b}$ symbol, we can recover the next vertex by specifying one of~$d-1$ neighbors of the current vertex; there are only $d-1$ possibilities, since~$\calH$ is non-backtracking.  (Exception: there are~$d$ choices at the very beginning of the hike; we compensated for this with the factor $2 > \frac{d}{d-1}$.)  To complete the proof of the claim, we need to show that for each stale stretch, there are at most~$2r\ell$ possibilities.  Recall that a stale stretch beginning from a vertex~$v$ consists of walking in non-backtracking fashion for at most~$r$ steps over the previously seen portion~$K$ of~$G_{\calH}$.  This subgraph~$K$ has at most $2\ell$ vertices, and by the bicycle-free property, this walk is confined to a subgraph of~$K$ that is at most unicyclic. It is easy to see this walk is determined by specifying its final vertex (at most~$2\ell$ possibilities), the number of times the cycle in $v$'s distance-$r$ neighborhood (should it exist) is traversed (fewer than $r/2$ possibilities), and the direction in which the cycle is traversed ($2$~possibilities).  Thus indeed each stale stretch can be completely determined by specifying one of at most $2\ell \cdot (r/2) \cdot 2 = 2r\ell$ possibilities.
\end{proof}

Combining this with \Cref{prop:tracemethod} now yields:
\begin{corollary}                                       \label{cor:itsover}
    Let $G = (V,E)$ be an arbitrary $d$-regular $n$-vertex graph.  Assume that $G$ is bicycle-free at radius~$r$.  Let $\ell \in \N^+$ and $0 < \eta < 1$ be parameters.  Then for $\bG$ a uniformly random edge-signing of~$G$, except with probability at most~$\eta$ the non-backtracking matrix~$\bB$ of~$\bG$ has spectral radius bound
    \begin{equation} \label[ineq]{eqn:tb}
        \rho(\bB) \leq \sqrt{d-1} \cdot \parens*{1 + O(\eps_1) + O(\eps_2)},
    \end{equation}
    where
    \[
        \eps_1 \coloneqq \frac{\log(n/\eta)}{\ell}, \qquad \eps_2 \coloneqq \frac{\log(d\ell)\log(\ell)}{r},
    \]
    provided $\eps_1, \eps_2 \leq 1$.

    Furthermore, if the random edge-signs of~$\bG$ are merely $(\delta,2\ell)$-wise uniform, the bound holds up to an additional additive~$(\delta n/\eta)^{\frac{1}{2\ell}} \cdot O(d)$.
\end{corollary}
\begin{proof}
    We have obtained that, for a uniformly random edge-signing $\bw : E \to \{\pm 1\}$,
    \[
        \E[\bT] \leq O(d^2 \ell^3 n) \cdot (d-1)^{\ell} \cdot (dr\ell)^{O\parens{\frac{\log \ell}{r}} \cdot \ell}.
    \]
    Note that $r \lesssim \log_{d-1} n$ always holds, and hence we must have $\ell \leq n$ (else $\eps_2 > 1$).  Also we must have $\ell \geq \log n$ (else $\eps_1 > 1$).  Thus we may coarsen $O(d^2 \ell^3 n)$ in the above to $O(n^5)$, and coarsen $(dr\ell)^{O(\cdot)}$ to  $(d\ell)^{O(\cdot)}$.  Now since $\bT$ is a nonnegative random variable, Markov's inequality implies that except with probability at most~$\eta$,
    \[
        \bT \leq O(n^5/\eta) \cdot  (d-1)^{\ell} \cdot (d\ell)^{O\parens{\frac{\log \ell}{r}} \cdot \ell},
    \]
    and hence
    \[
        \rho(\bB) \leq \bT^{\frac{1}{2\ell}} \leq O(n^5/\eta)^{\frac{1}{2\ell}} \cdot \sqrt{d-1} \cdot (d\ell)^{O\parens{\frac{\log \ell}{r}}},
    \]
    which directly implies \Cref{eqn:tb}.

    Finally, in the $(\delta,2\ell)$-wise uniform case, we get an additional additive $\delta n d^{2\ell+2}$ in the bound on~$\E[\bT]$; this gets a factor of~$1/\eta$ after the application of Markov, and becomes $(\delta n/\eta)^{\frac{1}{2\ell}} \cdot O(d)$ after taking $2\ell$th roots.
\end{proof}

Finally, the reader may verify that \Cref{thm:keyA} follows from \Cref{cor:itsover} in the fully uniform case by taking $\ell = \Theta(r \log(n)/\log \log n)$, and in the derandomized case by taking $\ell = \Theta(C \log(n/\eta))$.

\begin{remark}
    Alternatively, by taking $\eta = \exp(-\exp(r^{.49}))$ and $\ell = \exp(r^{.49})$ in \Cref{cor:itsover}, we may conclude that $\rho(\bB) \leq \sqrt{d-1}\cdot  (1 + o_r(1))$ holds in the fully uniform case except with probability at most $\exp(-\exp(r^{.49}))$.
\end{remark}

\section{Weakly derandomizing Bordenave's theorem} \label{sec:bordenave-derand}

In this section we give a weak derandomization of Bordenave's proof of \Cref{thm:friedman}, using ``off-the-shelf'' tools; the derandomization is ``weak'' in the sense that it only yields a $\operatorname{quasipoly}(n)$-time deterministic construction.  As discussed in \Cref{sec:rand-d-reg}, we will derandomize both the configuration model version and the random lift version. Specifically, we show the conclusion of  \Cref{thm:friedman} holds even for the ``almost $k$-wise uniform'' versions of these models, $k = O(\log n)$.
\begin{definition}[$(\delta,k)$-wise uniform configuration/lift models]
    When the permutation~$\bpi \in \symm{nd}$ used in the configuration model is not uniformly random but is merely $(\delta,k)$-wise uniform, we will say that $\bG$ is drawn from the \emph{$(\delta,k)$-wise uniform configuration model}.  Similarly, when the $\bpi_{uv} \in \symm{n}$ used in the random lift model are independent but merely $(\delta,k)$-wise uniform, we will say that $\bG$ is a \emph{$(\delta,k)$-wise uniform random $n$-lift of base graph $\ul{G}$}.
\end{definition}
\noindent (For simplicity, in the lift model we will henceforth only concern ourselves with  $\ul{G} = K_{d+1}$.)

We will not fully recap Bordenave's proof of \Cref{thm:friedman} in this work, although the reader unfamiliar with it will get some insight knowing that our proof of \Cref{thm:keyA} is modeled on it.  Bordenave employs two twists on the Trace Method to show that a random $d$-regular graph~$\bG$ has spectral radius at most $2\sqrt{d-1} + \eps$ (when the trivial eigenvalue of~$d$ is ignored).  The less important (but still challenging) twist involves replacing the non-backtracking matrix~$\bB$ by a centered variant,~$\ul{\bB}$, that enables one to ignore the trivial eigenvalue.  The more conceptually important twist comes from the fact, originally recognized by Friedman, that even after passing to~$\ul{\bB}$, the Trace Method still fails.  The reason, in brief, is as follows: A successful use of the Trace Method would have to consider walks of length~$\ell$ for $\ell$ at least a large multiple of $\log n$, in order to overcome the factor of~$n$ arising from the $n$ different walk starting points (cf.~the error term $\eps_1$ just after \Cref{eqn:tb}).  But for walks of this long length, one can show that the expected trace of $\ul{\bB}^\ell (\ul{\bB}^\top)^\ell$  is simply too large --- much larger than the target $\poly(n) \cdot (d-1)^{\ell}$ needed to get the ``correct'' final bound.

However, as first demonstrated by Friedman, the expectation is too large only because of certain low-probability events.
Bordenave's way of handling things is to show that: (i)~a random $d$-regular graph~$\bG$ is, with high probability, bicycle-free at large radius~$r$; (ii)~when $\bG$ is so bicycle-free, the $r$th power of its non-backtracking matrix, $\bB^r$, coincides with a certain ``bicycle-discarding'' variant~$\bB^{(r)}$; (iii)~the usual Trace Method \emph{can} be successfully applied to $\bB^{(r)}$; i.e., the expected trace of powers of $\ul{\bB}^{(r)}$ \emph{is} suitably small.

Thus our weak derandomization of Bordenave's proof has two ingredients, corresponding to (i)~and~(iii) above.  In \Cref{sec:bike-free} we derandomize a standard proof that a random $d$-regular graph is bicycle-free at large radius (in either the configuration model or the random lift model). In \Cref{sec:bordo} we examine  the key probabilistic ingredient in Bordenave's use of the Trace Method,  \cite[Prop.~11]{Bor19}, which encapsulates the fact that for a centered version~$\ul{\bM}$ of the configuration model matching matrix, the random variables $\ul{\bM}_{(v,i),(v',i')}$ are close to $k$-wise independent for $k \ll \sqrt{dn}$.

(In \Cref{sec:simple}, we also show a derandomization of the most basic fact about the configuration model, that~$\bG$ is simple with probability~$\Omega_d(1)$.  This is just a ``bonus'' for the reader who prefers the configuration model; it will be more convenient to use the random lift model for our explicit near-Ramanujan graphs, due to its guaranteed simplicity.)

\subsection{Bicycle-freeness}   \label{sec:bike-free}
The following relatively straightforward fact about $d$-regular $n$-vertex graphs is crucial for Bordenave's proof: with high probability they are bicycle-free at radius~$r$, provided $r \lesssim c \log_{d-1} n$ for some constant $c < 1/4$.   This fact is proved for completeness by Bordenave~\cite[Lem.~9]{Bor19} (and in \cite[Lem.~27]{Bor19} for random lifts); another proof appears earlier in, e.g.,~\cite[Lem.~2.1]{LS10}.  We would like a derandomized version of this fact for the $k$-wise uniform configuration model, $k = O(r)$.  This motivates looking for a moments-based proof, such as the one suggested by Wormald~\cite[Lem.~2.7]{Wor99} and carried out for Erd\H{o}s--R\'{e}nyi $\calG(n,m)$ graphs in~\cite[Thm.~5.5]{JLR00}.  The essential point will be that minimal witnesses to failure have only~$O(r)$ edges.
\begin{definition}[Minimal bicycle]
    We say a connected multigraph is a \emph{minimal bicycle} if it is bicyclic but has no proper subgraph that is bicyclic.  It is easy to see (cf.~\cite[Proof of Thm.~5.5]{JLR00}) that any minimal bicycle is either a ``handcuffs graph'' (two cycles joined by a path), a ``figure-eight graph'' (two cycles attached at a vertex), or a ``theta graph'' (a cycle with a ``diagonal'').
\end{definition}
We now prove:
\begin{proposition}                                     \label{prop:derand-bicycle-free}
    Fix  $d \geq 3$ and $k \geq 1$.  Let~$\bG$ be drawn from the $d$-regular $n$-vertex configuration model using a $2k$-wise uniform permutation.  Then $\bG$ is bicycle-free at radius~$k/4$, except with probability at most $O(k^3 (d-1)^{k}/n)$.  
    
    As a corollary, the failure probability is at most $1/n^{.99}$ provided $k < c \log_{d-1} n$ for a certain universal $c > 0$.  This statement remains true if $\bG$ is instead a $2k$-wise uniform random $n$-lift of~$K_{d+1}$.  Finally, by \Cref{thm:al13}, these statements remain true in the $(\delta,2k)$-wise uniform versions of the models, $\delta \leq 1/n^{8k+2}$.
\end{proposition}
\begin{proof}
    We first consider the configuration model.  Fix a minimal bicycle~$H$ with $h$ vertices and hence $h+1$ edges, where $h < k$.  Let the random variable $\bX_H$ denote the number of times that $H$ appears in~$\bG$. This is a polynomial of degree at most $h+1 \leq k$ in the entries of $\bG$'s adjacency matrix and hence a polynomial of degree at most~$2k$ in the permutation indicators $1[\bpi(j) = (v,i)]$.  Thus to compute $\E[\bX_H]$ we may assume $\bG$ is drawn from the usual configuration model (with a truly random permutation).  In this case, it is elementary to compute an exact formula for~$\E[\bX_H]$; as per~\cite[eqn.~(2.4)]{Bor16}, it is
    \begin{equation} \label{eqn:bord24}
        \E[\bX_H] = \frac{1}{bc} \frac{n(n-1)(n-2)\cdots(n-h+1)}{(nd-1)(nd-3)(nd-5) \cdots (nd-2h-1)}\prod_{u \in V(H)} d(d-1)\cdots(d-\deg_H(u)+1),
    \end{equation}
    where $b$ (respectively, $c$) is the number of edge- (respectively, vertex-)isomorphisms of~$H$.  For any minimal bicycle~$H$ we have $b \geq 1$, $c \geq 2$, and $\deg_H(u) \geq 2$ for all $u \in V(H)$.  The last of these facts implies the product on the right in \Cref{eqn:bord24} is at most $(d(d-1))^{h+1}$.  Also, the large fraction in the middle is asymptotic to $(d^{h+1} n)^{-1}$, and it is not hard to check it is always at most twice that.  Hence we conclude $\E[\bX_H] \leq (d-1)^{h+1}/n \leq (d-1)^k/n$.  Finally, it is easy to see that, up to isomorphism, the number of minimal bicycles with fewer than $k$ vertices is at most~$O(k^3)$.  Thus by Markov's inequality we conclude that the probability of having any minimal bicycle on fewer than~$k$ vertices is at most $k^3(d-1)^k/n$.  The claim about the configuration model now follows because any bicyclic radius-$k/4$ vertex neighborhood in~$\bG$ must contain a minimal bicycle with fewer than~$k$ vertices.  (The ``worst case'' is a figure-eight graph.)
    
    As for the model where $\bG$ is a $2k$-wise uniformly random $n$-lift of $K_{d+1}$, the proof is nearly identical. The only difference arises in the computation of $\E[\bX_H]$ --- instead of using an exact closed form expression for the quantity, one can elementarily upper bound $\E[\bX_H]$ by $O((d+1)^h/n)$ (assuming, say, $k \leq \sqrt{n}$).  From this slightly weaker bound, one can still draw the same conclusion that the failure probability is at most $1/n^{.99}$ for $k < c \log_{d-1} n$ (possibly with slightly smaller~$c$).
\end{proof}
%

\subsection{Bordenave's key probabilistic proposition}  \label{sec:bordo}

In this section we examine the last place in Bordenave's argument that uses randomness of the underlying graph~$\bG$; namely, \cite[Prop.~11]{Bor19} for the configuration model and \cite[Prop.~28]{Bor19} for the random lift model.  These nearly-identical propositions give an upper bound on a certain moment arising in his use of the Trace Method.  Unfortunately, the propositions are not as self-contained as the ones covered in \Cref{sec:bike-free}.  Rather than trying to give a complete summary of how Bordenave's argument works, we will proceed in a ``black-box'' fashion, only giving the bare minimum needed to verify derandomizability.  We refer the reader to~\cite{Bor19} for the complete picture.  As in~\cite{Bor19}, we will focus on the configuration model, and then describe the modifications necessary for the random lift model.

Here is the key probabilistic proposition (which can be viewed as a far more sophisticated version of \Cref{fact:easy}):
\begin{proposition}                                     \label{prop:11}
    (\cite[Prop.~11]{Bor19}.)  Let $\vec{E} = [n] \times [d]$, and let $\bM$ be a uniformly random matching on~$\vec{E}$ as in the configuration model \Cref{def:config-model}.  Also let $\ul{\bM}$ be the matrix obtained from $\bM$ by subtracting $\tfrac{1}{n'}$ from each entry, where $n' \coloneqq dn$.  Then for any $\gamma \in \vec{E}^{2k}$ with $1 \leq k \leq \sqrt{m}$ and any $0 \leq k_0 \leq k$, we have
    \begin{equation} \label[ineq]{eqn:bord11}
        \abs*{\E\bracks*{\prod_{t=1}^{k_0} \ul{\bM}_{\gamma_{2t-1}, \gamma_{2t}} \prod_{t = k_0+1}^k \bM_{\gamma_{2t-1}, \gamma_{2t}}}} \leq O\parens*{2^b \cdot (\tfrac{1}{n'})^a \cdot (\tfrac{3k}{\sqrt{n'}})^{a_1}}.
    \end{equation}
\end{proposition}
\noindent Here $a$, $b$, and $a_1$ on the right-hand side of \Cref{eqn:bord11} are certain quantities relating to the multiplicities of half-edges in~$\gamma$ and to~$k_0$. We omit these definitions here, as they won't be relevant for us.

Note that when $\bM$ is formed from a random permutation $\bpi$ on $[nd]$ as in \Cref{def:config-model}, each entry $\bM_{e,f}$ is a polynomial of degree~$2$ in the indicators $1[\bpi(j) = (v,i)]$.  It follows that the quantity inside the expectation in \Cref{eqn:bord11} is a polynomial of degree at most~$2k$ in these indicators.  We conclude:
\begin{corollary}                                       \label{cor:11}
    Let~$\bG$ be drawn from the $d$-regular $n$-vertex configuration model using a $2k$-wise uniform permutation, and write $\bM$ for the matching matrix inducing~$\bG$.  Then \Cref{eqn:bord11} continues to hold.
\end{corollary}

Bordenave also proved an analogue of \Cref{prop:11} for the random lift model.  The statement is extremely similar to \Cref{prop:11}, with ``$n'$'' being $n$, and with the rows/columns of ``$\bM$'' being the potential ``half-edges'' in the lifted graph; for the exact statement we refer the reader to \cite[Prop.~28]{Bor19}.  Further, \Cref{cor:11} is true when~$\bG$ is drawn from a $2k$-wise uniform lift model.

With \Cref{prop:11} in hand, Bordenave does some intricate --- but entirely non-probabilistic --- path-counting to complete his use of the Trace Method.  (This is like a much more sophisticated version of the part of \Cref{sec:random-signs} beginning with \Cref{def:hikes}.)  This part of his proof involves considering paths of length $2\ell m$, where ``$\ell$'' and ``$m$'' are parameters he selects (with $\ell$ being at least the bicycle-free radius, and~$m$ being large enough so that $\ell m \gg \log n$).  The crucial observation for us is that Bordenave \emph{only} employs \Cref{prop:11} with its parameter ``$k$'' set to~$2\ell m$ (and the same is true in the random lift model).

Bordenave directly sets $\ell = \Theta(\log_{d-1}(n))$ and $m = \Theta(\log(n)/\log \log(n))$ to obtain best parameters, but we will work more generally, since we may be interested in minimizing $k= 2\ell m$ to save on random bits.  Carefully examining \cite[Proofs of Prop.~14,~18]{Bor19}, one may extract the below proposition.  The random matrices $\ul{\bB}^{(\ell)}$ and $\bR^{(\ell)}_{1}, \dots, \bR^{(\ell)}_\ell$ mentioned in it are derived from the randomness of the configuration model; again, see~\cite{Bor19} for details.
\begin{proposition}                                     \label{prop:bordextract}
    Assuming $d$, $\ell$, $m$ satisfy $\poly(d \ell m)^{m} \ll n$, it holds that
    \[
        \E\bracks*{\|\ul{\bB}^{(\ell)}\|^{2m}} \leq \poly(n) \cdot (d-1)^{\ell m}, \qquad
        \E\bracks*{\sum_{i=1}^\ell \|\ul{\bR}_i^{(\ell)}\|^{2m}} \leq \poly(d \ell m)^m \cdot (d-1)^{2\ell m},
    \]
    Furthermore, this only relies on \Cref{eqn:bord11} with $k = 2\ell m$, and therefore by \Cref{cor:11} it continues to hold even in the $4\ell m$-wise independent configuration model.  Thus in this model, Markov's inequality implies that except with probability at most~$n^{-100}$,
    \[
        \|\ul{\bB}^{(\ell)}\| \leq \poly(n)^{\frac{1}{2m}} \cdot \sqrt{d-1}^{\ell}, \qquad
        \sum_{i=1}^\ell \|\ul{\bR}_i^{(\ell)}\| \leq \poly(n)^{\frac{1}{2m}} \cdot (d-1)^{\ell}.
    \]
\end{proposition}
This proposition holds just the same in the random lift model with base graph $\ul{G} = K_{d+1}$ (indeed, with any $d$-regular base graph).  One simply has to follow through the analogous propositions, \cite[Proofs of Prop.~29,~33]{Bor19}, in the same way.\footnote{Bordenave carries these propositions out for not-necessarily-regular base graphs of maximum degree~$d$.  His computations depend on the base graph through the Perron eigenvalue~$\rho_1$ of its non-backtracking operator~$B$, which in the $d$-regular case is just $d-1$.  In \cite[(67)]{Bor19} Bordenave selects $\rho > \rho_1$ and $c_\rho \geq 1$ such that $\|(B^\top)^k 1_e\|_1 \leq c_\rho \rho^k$ holds for all $k$ and all edges in the base graph.  In our $d$-regular case, we can simply take $c_\rho = 1$ and $\rho = \rho_1 = d-1$ when carrying through his computations.}

Finally, \cite[Prop.~8]{Bor19} is the following:
\begin{proposition}                                     \label{prop:8}
    Suppose $\bG$ drawn from the $d$-regular configuration model is bicycle-free at radius~$\ell$.  Let $n' = dn$. Then the largest magnitude eigenvalue of the associated non-backtracking matrix~$\bB$, excluding the trivial eigenvalue of~$d$, is at most
    \[
        \parens*{ \|\ul{\bB}^{(\ell)}\| + \frac{1}{n'} \cdot \sum_{i=1}^\ell \|\ul{\bR}_i^{(\ell)}\|}^{1/\ell}.
    \]
\end{proposition}
Again, Bordenave has a very similar analogue~\cite[Prop.~26]{Bor19} in the random lift model, with ``$n'$'' equal to~$n$, and with the quantity bounding the largest-in-magnitude ``new'' eigenvalue of the lifted graph (which is precisely what one needs to bound to show the near-Ramanujan property, assuming the base graph is itself $d$-regular Ramanujan).

We can now finish the proof as Bordenave does (in either the configuration or random lift model), combining \Cref{cor:IB}, \Cref{prop:8,prop:bordextract,prop:derand-bicycle-free}, and also \Cref{prop:simple-derand} (if desired).  Using the parameter settings $\ell = c \log_{d-1} n$ and $m = (C/c) \log(d-1)/\sqrt{\eps}$ where $c$ is the constant from \Cref{prop:derand-bicycle-free} and $C$ is a large enough universal constant, we get the following:
\begin{theorem}                                     \label{thm:derand-bord}
    Fix $3 \leq d \leq C^{-1}\sqrt{\log n}$ and let $\eps \leq 1$ and~$k$ satisfy  
    \[
        \eps \geq C^3 \cdot \parens*{\frac{\log \log n}{\log_{d-1} n}}^2, \qquad  k \geq  C \log(n)/\sqrt{\eps}.
    \]
    Let $\bG$ be chosen from the $d$-regular $n$-vertex $k$-wise uniform configuration model, or as a $k$-wise uniform random $n$-lift of $K_{d+1}$.     Then except with probability at most $1/n^{.99}$, the following hold:
    \begin{itemize}
        \item $\bG$ is bicycle-free at radius~$c \log_{d-1} n$;
        \item $\lambda(\bG) \leq 2\sqrt{d-1}\cdot(1+\eps)$.
    \end{itemize}
    Additionally, in the configuration model case, $\bG$ is simple with probability at least $e^{-(d-1)^2/4}/2$.  Finally, by \Cref{thm:al13}, these statements remains true in the $(\delta, k)$-wise uniform configuration model, $\delta \leq 1/n^{8k+1}$.
\end{theorem}
%

\ignore{
%
%
%

Recall that, as noted in \Cref{sec:iharabass}, one of the key steps to
obtain \Cref{thm:friedman} is to analyze the matrix~$\bB$, the
non-backtracking matrix of $\bG$. However, the definition of $\bB$ used
by Bordenave uses slightly different notation. To that effect, first
let $\bsigma$ to be a permutation associated to $\bpi$ defined in the
following way: for odd $j \in [nd]$ we set
$\bsigma(\bpi(j)) = \bpi(j + 1)$ and $\bsigma(\bpi(j + 1)) =
\bpi(j)$. Recall that both $\bpi$ and $\bsigma$ are permutations over
the set of half-edges $[nd]$ which we denote here by $\vec{E}$. Note
that $\bsigma(\bsigma((v, i))) = (v, i)$ for all $(v, i) \in
\vec{E}$. The intuition is that $(u, i) \in \vec{E}$ is the $i$th edge
that leaves vertex $u$ and goes to $v$, where
$\bsigma((u, i)) = (v, j)$. So Bordenave defines $\bB$ by, for
$e = (u, i), f = (v, j) \in \vec{E}$:

\[
  \bB_{ef} = \sum_{x \in [d] \setminus \{j\}} 1[\bsigma(e) = (v, x)].
\]

Let us now look at some more notation that will be useful later
on. Let $\bM$ be the permutation matrix associated to $\bsigma$; that
is, for all $e,f \in \vec{E}$:

\[
  \bM_{ef} = \bM_{fe} = 1[\bsigma(e) = f].
\]

Also, let $N$ be the matrix defined by, for all
$e = (u, i), f = (v, j) \in \vec{E}$:

\[
  N_{ef} = N_{fe} = 1[v = u] \cdot 1[i \neq j],
\]

which means $N$ is the non-backtracking indicator matrix, more precisely
given two edges $e, f \in \vec{E}$, they form a non-backtracking two
path if $N_{ef} = 1$. Finally, observe that $\bB = \bM N$. To see why this
notation is useful, denote by $\vec{E}^\ell$ the set of sequences of $\ell$
half-edges. Similarly, let $\Gamma^\ell$ be the set of sequences of $\ell$
half-edges that form a non-backtracking walk. Let $\vec{E}^\ell_{ef}$ be
the subset of $\vec{E}^\ell$ whose first half-edge is $e$ and last is $f$
and define $\Gamma^\ell_{ef}$ analogously. Then we have:

\[
  (B^\ell)_{ef} = \sum_{\gamma \in \vec{E}^{2\ell+1}_{ef}}\prod_{s=1}^\ell \bM_{\gamma_{2s-1}\gamma_{2s}} N_{\gamma_{2s-1}\gamma_{2s+1}} = \sum_{\gamma \in \Gamma^{2\ell+1}_{ef}}\prod_{s=1}^\ell \bM_{\gamma_{2s-1}\gamma_{2s}}.
\]

In order to apply the trace method to bound the eigenvalues of
$B^\ell$ it is required to first get rid of the highest eigenvalue,
the trivial eigenvalue of $(d-1)^\ell$. To do so, first let
$\ul{\bM}_{ef} = \bM_{ef} - \frac{1}{dn}$, for all $e, f \in \vec{E}$,
which is a centered version of $\bM$. With it, we can similarly define
$\ul{\bB} = \ul{\bM}N$ or alternatively:

\[
  \ul{\bB} = \sum_{\gamma \in \Gamma^{2\ell+1}_{ef}}\prod_{s=1}^\ell \ul{\bM}_{\gamma_{2s-1}\gamma_{2s}}.
\]

Using these notions, Bordenave shows something like the
following\pnote{This is not really true because I skipped the tangle
  free stuff}:

\[
  |\lambda_2(\bB)| \leq \left (||\ul{\bB^l}|| + \frac{1}{dn}\sum_{i=1}^\ell ||\bR_i^\ell|| \right)^{1/\ell},
\]

where $\bR$ is a ``remainder'' matrix that only depends on $\bM$ and
$\ul{\bM}$. The core of Bordenave's work is now on how to bound these
two terms. To achieve that, Bordenave proves the following
proposition:

\begin{proposition}[\cite{Bor19} Proposition 11]
  \label{prop:bord11}
  Fix $d \geq 3$ and $1 \leq k \leq \sqrt{dn}$. Let $\bG$ be drawn from the $d$-regular $n$-vertex configuration model using a fully uniform permutation. Then there exists a universal constant $c$ such that for $\gamma \in \vec{E}^{2k}$ and any $0 \leq k_0 \leq k$ we have:
  \[
    \left | \E \left [ \prod_{t=1}^{k_0} \ul{\bM}_{\gamma_{2t-1}\gamma_{2t}} \prod_{t=k_0+1}^{k} \bM_{\gamma_{2t-1}\gamma_{2t}} \right ] \right | \leq c 2^b \left( \frac{1}{dn}\right )^a \left( \frac{3k}{\sqrt{dn}}\right )^{a_1},
  \]
  where $a, b, a_1$ are variables on $\gamma$ and $k_0$ \footnote{You can find the formal definition of these variables in \cite{Bor19}, but we only need them in a blackbox fashion for our purposes.}.
\end{proposition}

\begin{remark}
  When applying \Cref{prop:bord11} on Bordenave's remaining argument,
  the value of $k$ considered is always
  $O \left (\frac{\log^2 n}{\log\log n} \right)$.
\end{remark}

This proposition, along with the non-derandomized results of
\Cref{sec:simple} and \Cref{sec:bike-free} let Bordenave conclude the
following:

\begin{theorem}
  \label{prop:nonnackbord}
  Let $d \geq 3$, then if $\bG$ is drawn from the $d$-regular
  $n$-vertex fully uniform configuration model, then for any $a > 0$ there exists $c > 0$:

  \[
    \Pr[|\lambda_2(\bB)| \geq \sqrt{d-1} + \epsilon] \leq n^{-a},
  \]

  for $\epsilon = c \frac{\log\log n}{\log n}$, where $c = 8\log(d)(13 + \frac{13}{2} a)$ \pnote{After a deep dive into Bordenave this is the constant I found, we probably don't need to have its exact value in our paper}.
\end{theorem}

We finally make the observation that the above
is derandomizable:

\begin{proposition}
  \label{prop:debord11}
  \Cref{prop:bord11} remains true if $\bG$ is drawn from the
  $d$-regular $n$-vertex configuration model using a $2k$-wise uniform
  permutation.

  By \Cref{thm:al13}, this remains true if the permutation is $(\delta,2k)$-wise uniform, $\delta \leq n^{-O(k)}$.
\end{proposition}
\begin{proof}
  This follows almost immediately by definition. Observe that the
  quantity

  \[ \prod_{t=1}^{k_0} \ul{\bM}_{\gamma_{2t-1}\gamma_{2t}}
  \prod_{t=k_0+1}^{k} \bM_{\gamma_{2t-1}\gamma_{2t}}
  \]

  is just a polynomial on the values of $\bM$. Furthermore, this is a degree $k$
  polynomial. In turn, from the definitions of $\bM$ and $\bsigma$, we
  can expand this expression and conclude it is a polynomial of degree
  at most~$2k$ in the permutation indicators $1[\bpi(j) =
  (v,i)]$. Hence, we can assume the permutation in \Cref{prop:bord11}
  is $2k$-wise uniform and the result still stands.
\end{proof}

\begin{corollary}
  \Cref{prop:nonnackbord} remains true if the permutation is
  $(\delta,2k)$-wise uniform, $\delta \leq n^{-O(k)}$, for $k = O \left (\frac{\log^2 n}{\log\log n} \right)$.
\end{corollary}

\subsection{Summary of derandomizing Bordenave} \label{sec:summary-bord}

Combining all results in this section,\rnote{Likely details need to be added} we obtain:\rnote{check!!  I half made this up}
\begin{theorem}                                     \label{thm:derand-bord}
    Let $3 \leq d \ll \sqrt{\log n}$ and let $1 \leq \wt{m} \ll \log_{d-1}(n)/\log \log n$.\rnote{In my mind, Bordenave's $m$ is $\wt{m} \log d$. I'm trying to work out all the $d$-dependencies, just for Sid :)}  
    Let~$\bG$ be drawn from the $d$-regular $n$-vertex configuration model using a $k$-wise uniform permutation, $k = O(\wt{m} \log n)$.  Then with probability at least~$\exp(-O(d^2))$:
    \begin{itemize}
        \item  $\bG$ is simple;
        \item $\bG$ is bicycle-free at radius $r = c \log_{d-1} n$, where $c > 0$ is a universal constant;
        \item excluding the trivial eigenvalue of~$d$, all of $\bG$'s eigenvalues are at most $2\sqrt{d-1} + \eps$ in magnitude, where $\eps = O(\wt{m}^{-2})$.
    \end{itemize}
    By \Cref{thm:al13}, this remains true if the permutation is merely $(\delta,k)$-wise uniform, $\delta \leq n^{-O(k)}$.
\end{theorem}
}
\section{Explicit near-Ramanujan graphs} \label{sec:final-construction}

With the tools developed in \Cref{sec:random-signs} and
\Cref{sec:bordenave-derand} we are now ready to establish our explicit near-Ramanujan graph constructions.  For ease of reading, in this section we will merely prove \Cref{thm:us}, the  deterministic polynomial-time (``weakly explicit'') construction, with $d$ and $\eps$ assumed to be constant.  We leave the slightly more technical proof of the ``probabilistically $\poly\log n$-time computable'' construction (\Cref{thm:prob-strong-explicit}), with worked out dependence on $d = d(n)$ and $\eps = \eps(n)$, for \Cref{sec:prob-strong-explicit}.



Recall we want to show there is a deterministic  algorithm that on input $N, d \geq 3$ and $\eps > 0$, outputs in $\poly(N)$-time a $d$-regular graph~$G$ on $N' \sim N$ vertices with $\lambda(G) \leq 2\sqrt{d-1} + \eps$.

Before getting into the details, we recap the construction as outlined in \Cref{sec:ideas}:

\begin{enumerate}
    \item Using \Cref{thm:derand-bord} we construct a $d$-regular simple graph~$G_0$ on some ``small'' number of vertices $n_0 = n_0(N)$, which is bicycle-free at radius $\Omega(\log n_0)$ and has  $\lambda(G_0) \leq 2\sqrt{d - 1} + \epsilon$.  The quantity~$n_0$ should satisfy
          \[
            2^{\omega((\log \log N)^2)} \leq n_0 \leq 2^{O(\sqrt{\log N})},
          \]
          the left inequality so that $G_0$ is sufficiently bicycle-free for Step~2 below, and the right inequality so that $G_0$ is constructible in deterministic $\poly(N)$ time.  We have a wide range of allowable possibilities here; for concreteness we  will take $n_0$ near the upper limit to allow for slightly better dependence on non-constant $d$, $\eps$ in \Cref{sec:prob-strong-explicit}.
\item Next we repeatedly use \Cref{thm:keyA} (roughly $\log (N/n_0) \sim \log N$ times) to double the number of vertices in our construction from Step~1, while keeping $\lambda \leq 2\sqrt{d - 1} + \epsilon$ and also retaining that the graph is bicycle-free at radius $\Omega(\log n_0)$ (\Cref{prop:lift-is-bike-free}).  Importantly, since \Cref{thm:keyA} is a high-probability result, we will be able to reuse the seed for each of the $\log N$ pseudorandom edge-signings.
\end{enumerate}

\paragraph{Step 1 details.} Here the algorithm will select $n_0$ to be an even integer on the order of $2^{\Theta(\sqrt{\log N})}$.  \Cref{thm:derand-bord} tells us that for a sufficiently large $k = O(\log n_0) = O(\sqrt{\log N})$, and for sufficiently small $\delta = n_0^{-\Theta(k)} = 1/\poly(N)$, a random $d$-regular $n_0$-vertex graph $\bG_0$ chosen from the $(\delta,k)$-wise uniform configuration/random-lift-of-$K_{d+1}$ model will with high probability satisfy:
\begin{equation}    \label{eqn:needed1}
    \bG_0 \text{ is bicycle-free at radius } \Omega(\log n_0) = \Omega(\sqrt{\log N}); \qquad \lambda(\bG_0) \leq 2\sqrt{d-1} + \eps.
\end{equation}
(Recall we are treating $d$ and $\eps$ as constant here.)  $\bG_0$ will also be simple with $\Omega(1)$ probability in the configuration model case, and with probability~$1$ in the random lift case.  In the former case, we need a $(\delta,k)$-wise permutation in $\symm{nd}$; in the latter case, we need $\binom{d+1}{2}$ independent $(\delta,k)$-wise permutations in $\symm{n}$.  Either way, \Cref{thm:knr} tells us that a deterministic algorithm can enumerate all possibilities for $\bG_0$ in $\poly(N)$ time and pick out any fixed simple one $G_0$ satisfying~\eqref{eqn:needed1}.

\paragraph{Step 2 details.} Here the algorithm will be applying \Cref{thm:keyA} some $t \sim \log_2 N$ times, starting with $G_0$, and each time interpreting the edge-signing produced as a $2$-lift as discussed in \Cref{sec:ideas}.  This produces a sequence of pseudorandom $d$-regular simple graphs $\bG_1, \dots, \bG_t$, where $\bG_i$ has $n_0 2^t$ vertices. The parameter $t$ is chosen to be least possible such that the final number of vertices, $N' = n_0 2^t$, is as at least~$N$.  It is not hard to check that by adjusting $n_0$ by a factor of at most~$2$, we can ensure that $N'/N = 1 + o_N(1)$, where the $o_N(1)$ term is $O(1/n_0) = 1/2^{\Theta(\sqrt{\log N})}$.  

For simplicity, we will use the same values for the parameters $r$, $k$, and $\delta$ in each application of  \Cref{thm:keyA}; only the value of $n$ will change (ranging from $n_0$ up to~$N'$).  We may take $r = \Omega(\sqrt{\log N})$, the bicycle-free radius from~\Cref{eqn:needed1} (observe that the bicycle-free radius cannot decrease for \emph{any} $2$-lift of a graph).  Note that the failure probability of any single $2$-lift is at most $1/2^{\Theta(\sqrt{\log N})}$, and hence a union bound tells us that the probability of \emph{any} of the $2$-lifts ``failing'' is low, $\frac{\log N}{2^{\Theta(\sqrt{\log N})}}$.  We take the parameter ``$k$'' to be $\Theta \left(\frac{\log N}{\sqrt{\eps}} \right)$  (the hidden constant sufficiently large depending on~$d$).  Finally, we take $\delta = 1/N^{\Theta(1/\sqrt{\eps})}$ (again with the hidden constant sufficiently large depending on~$d$).
By plugging these parameters into \Cref{thm:keyA} we conclude that with high probability, all ``new'' eigenvalues arising in the $2$-lifted adjacency matrices $\bA_1, \dots, \bA_t$ are at most $2\sqrt{d-1} + \eps$ in magnitude, and hence $\bG_t$ is $\eps$-near Ramanujan.

It remains to observe that with these parameter settings, using \Cref{thm:nn93}, a deterministic algorithm can in $\poly(N/\delta) = \poly(N)$ time do the following:  First, produce a single $(\delta,2\ell)$-wise uniform multiset of strings~$Y \subseteq \{\pm 1\}^{N'd/4}$; here $N'd/4$ bits are sufficient to edge-sign/$2$-lift any of the graphs $\bG_i$.  Then, for $i=1,\dots,t$ the algorithm can search $Y$ for a ``good'' string $y_i \in Y$, meaning one with the property that using it to do an edge-signing/$2$-lift of $G_i$ yields graph $G_{i+1}$ which is $\eps$-near Ramanujan.   As argued in the previous paragraph, a $1-O\left(\frac{\log N}{2^{\Theta(\sqrt{\log N})}}\right)$ fraction of strings in $Y$ have this property.  We can check the goodness of any string~$y$ in $\poly(N)$ time using the following fact.
\begin{fact} \label{fact:check}
    For any rational approximation~$\rho$ of $2\sqrt{d-1} + \eps$, one can decide in $\poly(n)$ time whether $\lambda(G) \leq \rho$.
\end{fact}

This concludes the proof of \Cref{thm:us}.

\section*{Acknowledgments}
We would like to thank Michael Forbes and Zander Kelley for pointing out an error in an earlier version of the paper.

\bibliographystyle{alpha}
\bibliography{random-graph}

\appendix
\section{Simplicity} \label{sec:simple}
In the fully uniform configuration model, the probability of~$\bG$ being simple (i.e., being an ordinary graph with no self-loops or parallel edges) is known~\cite{BC78,Bol80} to tend to the constant \mbox{$\exp(-(d^2-1)/4)$}, as $n \to \infty$.  We establish that the $O(d^2)$-wise uniform configuration model suffices for this:
\begin{proposition}                                     \label{prop:simple-derand}
    Let $3 \leq d \ll \sqrt{\log n}$ and let $k \geq Cd^2$, where $C$ is a certain universal constant.  Let~$\bG$ be drawn from the $d$-regular $n$-vertex configuration model using a $k$-wise uniform permutation.  Then
    \[
        \Pr[\bG \textnormal{ is simple}] = e^{-(d^2-1)/4}(1 \pm e^{-100d^2}).
    \]
    By \Cref{thm:al13}, this remains true if the permutation is merely $(\delta,k)$-wise uniform, $\delta \leq n^{-C'd^2}$.
\end{proposition}
The proof is a straightforward derandomization of Bollob\'as's original analysis of simplicity in the configuration model~\cite{Bol80}.  Unlike several later refinements that used the Chen--Stein method, Bollob\'{a}s's proof uses the method of moments, making it particularly convenient to derandomize using $k$-wise uniform permutations.
\begin{proof} [Proof of \Cref{prop:simple-derand}]
Let us recap Bollob\'{a}s's proof concerning an $n$-vertex $d$-regular configuration model graph~$\bG$ formed from a truly random permutation~$\bpi \sim \symm{nd}$.  He defines $\bX_1$ to be the number of self-loops in~$\bG$ (i.e., $\tfrac12 \tr(\bA)$), $\bX_2$ to be the number of $2$-cycles (i.e., $\sum_{v < v'} \binom{\bA_{v,v'}}{2}$), and $\bX = \bX_1 + \bX_2$.  Note that $\bG$ is simple if and only if~$\bX = 0$. The idea of the proof is that it is nearly the case that $\bX_1$, $\bX_2$ are independent Poisson random variables with respective means
\[
    \lambda_1 =  \lambda \cdot \frac{nd}{nd-1} = \lambda \cdot (1 \pm O(1/n)), \qquad
    \lambda_2 =  \lambda^2 \cdot \frac{nd \cdot (nd-d)}{(nd-1)\cdot (nd-3)} = \lambda^2 \cdot (1\pm O(1/n)),
\]
where $\lambda \coloneqq (d-1)/2$.  Thus $\bX$ should be nearly Poisson with mean $\lambda_1 + \lambda_2 \sim \lambda + \lambda^2 = (d^2-1)/4$, and hence we should have $\Pr[\bX = 0] \sim e^{-(d^2-1)/4}$.

More precisely, Bollob\'{a}s first establishes~\cite[ineq.~(11)]{Bol80} the following estimate for all integers $0 \leq r \leq 8 \log n$:
\begin{equation}    \label[ineq]{eqn:simple-mom}
    E_r \coloneqq \E\bracks*{\binom{\bX}{r}} \text{ satisfies } \abs*{E_r - \frac{(\lambda_1+\lambda_2)^r}{r!}} \leq \frac{(\lambda_1+\lambda_2)^r}{r!} \cdot O(r^2/n).
\end{equation}
(Actually, Bollob\'{a}s has $O((\log n)^2/n)$ on the right-hand side rather than $O(r^2/n)$, but inspection of his proof confirms the above.)  The key point for our proof of \Cref{prop:simple-derand} is that \Cref{eqn:simple-mom} continues holds when the permutation $\bpi \in \symm{nd}$ defining~$\bG$ is merely $4r$-wise uniform.  This is simply because $\binom{\bX}{r}$ is a polynomial of degree at most~$4r$ in the indicators $1[\bpi(j) = (v,i)]$.  Thus to complete the proof, it suffices to derive the conclusion
\begin{equation}    \label{eqn:p0}
    \Pr[\bX = 0] = e^{-(d^2-1)/4}(1 \pm e^{-100d^2})
\end{equation}
from the estimates in \Cref{eqn:simple-mom} with $r = O(d^2)$.  This can be done exactly as in Bollob\'{a}s's work.  He uses the following inclusion-exclusion-type inequality, which holds (for any $u \in \N$) due to $\bX$ being $\N$-valued:
\begin{equation}    \label[ineq]{eqn:lhsrhs}
    \sum_{r=0}^{2u+1} (-1)^r E_r \leq \Pr[\bX = 0] \leq \sum_{r=0}^{2u} (-1)^r E_r.
\end{equation}
Notice that $E_r \approx \frac{(\lambda_1+\lambda_2)^r}{r!}$, and
\begin{equation}    \label{eqn:exact}
    \sum_{r=0}^\infty (-1)^r \frac{(\lambda_1+\lambda_2)^r}{r!} = e^{-(\lambda_1 + \lambda_2)} = e^{-(\lambda + \lambda^2)\cdot (1\pm O(1/n))} = e^{-(d^2 - 1)/4}\cdot (1 \pm O(d^2/n));
\end{equation}
also, $O(d^2/n) \ll e^{-100d^2}$ since $d \ll \sqrt{\log n}$.  Thus we can establish \Cref{eqn:p0} by bounding the two errors distinguishing the infinite sum in  \Cref{eqn:exact} from the sums on the left- and right-hand side of \Cref{eqn:lhsrhs}.  The two distinctions are: the error in $E_r \approx \frac{(\lambda_1+\lambda_2)^r}{r!}$, boundable using \Cref{eqn:simple-mom}; and, the tail of the infinite sum from $2u$ or $2u+1$ onward.  In absolute value, these two errors are boundable by:
\[
    O(1/n) \cdot \sum_{r=0}^{2u \text{ or } 2u+1} \frac{(\lambda_1+\lambda_2)^r}{r!} \cdot r^2, \quad \text{and} \quad \sum_{r=2u+1 \text{ or } 2u+2}^{\infty}\frac{(\lambda_1+\lambda_2)^r}{r!}.
\]
The first quantity above can be bounded by $O(u^2/n) \cdot e^{\lambda_1 + \lambda_2}$, and the second quantity can be bounded by $O((\lambda_1 + \lambda_2)^{2u+1}/(2u+1)!)$ provided $u \geq \lambda_1+\lambda_2$.  Recalling $\lambda_1 + \lambda_2 = \Theta(d^2)$ and $d \ll \sqrt{\log n}$, we see that by taking $u = O(d^2)$ sufficiently large, both errors can be made much smaller than $e^{-100d^2}$, and we obtain \Cref{eqn:p0} with $r = O(d^2)$ as needed.
\end{proof}

\section{The probabilistically $\poly\log n$-time computable construction}
\label{sec:prob-strong-explicit}

We now walk through the steps of \Cref{sec:final-construction} giving precise parameter details along the way, and extract a probabilistically $\poly\log n$-time computable construction of near-Ramanujan graphs.

Assume we are given $N$, $3\le d \le \frac{\left(\log N\right)^{1/8}}{C}$ and $\epsilon \gg \frac{(\log\log N)^4}{\log N}\cdot\sqrt{d}$ where $C$ is the constant from the statement of \Cref{thm:derand-bord}.

\paragraph{Revisiting Step 1.}  Choose parameters as follows: $\alpha=1/\sqrt{d+1\choose 2}$; $n_0$ as the largest multiple of $d+1$ smaller than $2^{\alpha\sqrt{\log N}}$; $k=C\alpha\sqrt{\log N}\cdot d^{1/4}/\sqrt{\epsilon}$ (which is $\approx \log n_0$); and $\delta=1/N^{8k+1}$.  Recall that the key result used in this step is that by \Cref{thm:derand-bord}, $\bG_0$ drawn from the $n_0$-vertex $(\delta,k)$-wise random-lift-of-$K_{d+1}$ model is a simple graph that with high probability satisfies:
\begin{align}
        \bG_0 \text{ is bicycle-free at radius } \Omega\parens*{\frac{\alpha\sqrt{\log N}}{\log(d-1)}}; \qquad \lambda(\bG_0) \leq 2\sqrt{d-1} + \eps.
\end{align}
As an upshot of \Cref{thm:knr}, $\bG_0$ can be sampled using $\bs$, a uniform binary string of length ${O\parens*{\frac{\log N\cdot d^{1/4}}{\sqrt{\epsilon}}}}$ as a seed.  In particular, $\bs$ is divided into ${d+1\choose 2}$ disjoint substrings $\bs_{e_1},\dots,\bs_{e_{d+1\choose 2}}$ each of length $\ell_1=O\left(\frac{\alpha^2\log N\cdot d^{1/4}}{\sqrt{\epsilon}}\right)$ indexed by edges of $K_{d+1}$;  the $(\delta,k)$-wise uniform permutation $\bpi_{uv}$ corresponding to edge $(u,v)$ is taken to be the $\bs_{uv}$th permutation in the multiset of permutations~$\Pi$ from the statement of \Cref{thm:knr}.  Additionally, given $\bs$ and a vertex $(u,i)\in V(\bG_0)$, it is possible to return a list of its neighbors in time $T_1=O\parens*{d \cdot \poly\parens*{\frac{\alpha^2 \log N \cdot d^{1/4}}{\sqrt{\epsilon}}}}$.

\paragraph{Revisiting Step 2.}  Let $t=\ceil*{\log\parens*{\frac{N}{n_0}}}$; let $\beta$ be a large enough constant; let $k=\frac{2\beta d^{1/4}}{\sqrt{\epsilon}}\log N$; and let $\delta=N^{-O(\beta d^{1/4}\log d/\sqrt{\epsilon})}$.  The main result used in Step 2 is that from \Cref{thm:keyA} the graphs $\bG_{1},\dots,\bG_t$ where $\bG_{i}$ is obtained via a $2$-lift of $\bG_{i-1}$ induced by a $(\delta,k)$-wise uniform signing have their nontrivial eigenvalues bounded by $2\sqrt{d-1}+\epsilon$ in magnitude, except with probability $O(t/n_0^{100})$.  From \Cref{thm:nn93}, a $(\delta,k)$-wise uniform signing of any $\bG_i$ can be obtained by first sampling a random binary string $\bs'$ of length $\ell_2=O\left(\frac{d^{1/4}\log d\cdot\log N}{\sqrt{\epsilon}}\right)$ and choosing the $\bs$th string in the multiset of signings $Y$ from the theorem statement.  In fact, given $\bs'$ and edge $e\in\bG_i$ one can also output the sign assigned to edge $e$ in time $T_2=\poly\left(\beta d^{1/4}\log d\log N/\sqrt{\epsilon}\right)$.  Finally, by the union bound, the bound of $O(t/n_0^{100})$ on the probability that $\bG_t$ is not $\epsilon$-near Ramanujan holds if we use independently chosen seeds $\bs_1,\dots,\bs_t$ to perform the $2$-lifts.  Note that $t<\log N$ and $n_0\ge 2^{(\log N)^{1/4}}$ and hence the failure probability is $o_N(1)$.

\paragraph{Probabilistically strongly explicit near-Ramanujan graphs.}  Given a uniform binary string $\bs$ of length $\ell_1+t\cdot\ell_2$ as a random seed, call the substring given by the first $\ell_1$ bits $\bs_1$ and the substring given by the next $t\cdot\ell_2$ bits $\bs_2$.  Let $\bG_0$ be sampled from $\bs_1$ as described in Step 1, and let $\bG_t$ be the ``final graph'' obtained by the sequence of $2$-lifts in Step 2 from $\bs_2$.  Each vertex in $\bG_i$ can be naturally identified with a tuple $(v,a,x) \in [d]\times[n_0]\times\{0,1\}^{i}$.  Let $x$ be a string in $\{0,1\}^t$, let $x^{\le i}$ denote its $i$-bit prefix.  Given a vertex $(v,a,x)$ in $\bG_t$ and seeds $\bs_1$ and $\bs_2$, we describe an algorithm to output a list of its $d$ neighbors in $\wt{O}(T_1+dT_2)$-time where the $\wt{O}(\cdot)$ hides factors of $\polylog N$.  From Step 1, we know that there is an $T_1$-time algorithm to output a list of $d$ neighbors of $(v,a,x^{\le 0})$ in $\bG_0$.

Next, given a list of neighbors of $(v,a,x^{\le i-1})$ in $\bG_{i-1}$ it is possible to output a list of neighbors of $(v,a,x^{\le i})$ in $\bG_i$ in $\wt{O}(dT_2)$-time in the following way.  Let $(w,b,y)$ be a neighbor of $(v,a,x^{\le i-1})$.  Then exactly one of $(w,b,y\wedge 0)$ and $(w,b,y\wedge 1)$ is a neighbor of $(v,a,x^{\le i})$ where $\wedge$ denotes concatenation.  It is possible to obtain the sign on edge $\{(v,a,x^{\le i-1}),(w,b,y)\}$ in the $2$-lift from $\bG_{i-1}$ to $\bG_i$ in $T_2$ time from $\bs_2$.  If the sign is a $-1$, then $(w,b,y\wedge(1-x_i))$ is a neighbor of $(v,a,x^{\le i})$; otherwise $(w,b,y\wedge x_i)$ is a neighbor.  Thus, in $\wt{O}(dT_2)$ time, we can obtain a length-$d$ (and hence complete) list of neighbors of $(v,a,x^{\le i})$.

As a result, after spending $T_1$ time generating a list of neighbors of $(v,a,x^{\le 0})$, we can use the above routine $t$ times to obtain a list of neighbors of $(v,a,x)$ in $\bG_t$ in $T_1+t\cdot\wt{O}(dT_2) \le \wt{O}(T_1+dT_2)$.  From the upper and lower bounds on $d$ and $\epsilon$, this quantity is always $O(\polylog N)$.

To summarize, we have an algorithm that takes in a random seed of length $O\left(\frac{d^{1/4}\log d\cdot\log^2 N}{\sqrt{\epsilon}}\right)$ and implements the adjacency matrix of a corresponding random graph $\bG$ such that:
\begin{itemize}
    \item Given any vertex $v$ of $\bG$, its list of neighbors can be generated in $O(\polylog N)$ time.
    \item $\bG$ is $\epsilon$-near Ramanujan with probability $1-o_N(1)$.
\end{itemize}
This yields the conclusion of \Cref{thm:prob-strong-explicit}.

\ignore{
\subsection{Older material written by Pedro}

\paragraph{Step 2 details.}

Let's first detail the first step of this construction. We need a
$(m^{-8Cd^{1/4}\log(m)/\sqrt{\epsilon}},
C\log(m)d^{1/4}/\sqrt{\epsilon})$-wise uniform permutation of length $dm$
with $m = 2^{(\log \log N)^4}$, to apply theorem
\Cref{thm:derand-bord} and produce a simple graph with $m$ vertices,
which succeeds with probability at least $\exp(−O(d^2))$\pnote{I'm
  ignoring how $\epsilon$ has to be at least
  $C \cdot (\log\log m / \log_{d-1} m)^2$, but we can easily control
  this by choosing something like
  $m = 2^{(1/\epsilon) (\log \log N)^4}$. I might change this, but it
  will introduce extra technicalities}. From \Cref{thm:knr}, this
takes time
$\poly(m^{d^{1/4} \log(m) / \sqrt{\epsilon}}) = 2^{\poly(\log \log N)}
= o(\poly(N))$ to construct and requires a seed of length
$O((d^{1/4} \log(m) / \sqrt{\epsilon}) \log(m)) = O((\log\log N)^8) =
o(\log N)$. Hence, we can safely enumerate over all possible seeds in
sub-polynomial time and a constant fraction (in particular a non-zero
fraction) will satisfy the desired conditions. Recall
\Cref{fact:check}, which tells us we can safely check $\lambda$ is
small enough in $\poly(N)$ time (also note that checking
$r$-bicycle-freeness for any $r$ and simplicity is trivially done in
$\poly(N)$ time).

As a result of the above, we get a $d$-regular graph with
$\lambda \leq 2\sqrt{d - 1} + \epsilon$. Let's denote this graph by
$G_0$. Also, observe that $G_0$ is bicycle-free at radius
$O \left(\log \left( 2^{(\log \log N)^4} \right) \right) = O((\log
\log N)^4) \gg (\log\log N)^2$, so we can apply \Cref{thm:keyA} on the
next step.

Now on the second step of the construction we are iteratively applying
\Cref{thm:keyA} to boost the number of vertices to the required
amount. Let's denote by $G_i$ the graph obtained after the $i$th
iteration, so that $G_i$ has $|V_i| = 2^i m$ vertices. First of all,
since each iteration doubles the number of vertices in the graph, if
we perform it
$\left \lceil \log N - (\log \log N)^4 \right \rceil$ times, we
get a graph of size
$\left \lceil 2^{(\log \log N)^4} \right \rceil \cdot \left \lceil
  2^{\log N - (\log \log N)^4} \right \rceil = N(1 + 1/(\log \log
N)^4) = N(1 + o(1))$. Furthermore, observe that given that $G_0$ is
$(\log\log N)^4$-bicycle-free then all the $G_i$ will also be, since
they form a chain of 2-lifts. Given this setup, we have to prove two
facts: that this can be done in polynomial time; that this provides
the right $\lambda$ bound.

\textbf{Proof of polynomial time:} To produce $G_i$ we find a good
2-lift of a graph ($G_{i - 1}$). So, we need a binary string of length
$d 2^{i - 2} m$ of $\left (N^{-O(C \log d)}, 2C \log \left (2^{i - 1}
m \right)
\right)$-wise uniform bits (for some $C = C(N)$ to be picked later),
in order to apply \Cref{thm:keyA}, which succeeds with probability at
least $1 - (2^{i - 1} m)^{-100}$. According to \Cref{thm:nn93} this
takes time
$\poly(d 2^{i - 2} m \cdot N^{O(C \log d)}) = \poly(N^{C \log d})$ to
construct (note that $2^{i - 2} m = O(N)$) and requires a seed of
length $O((C \log d) \log N)$ per iteration. Hence, we can safely
enumerate over all possible seeds in time $\poly(N^{C \log d})$ and a
$1 - o(1)$ fraction (in particular a non-zero fraction) will satisfy
the desired conditions. Since we do at most $O(\log N)$ iterations the
total running time of the full construction is still
$\poly(N^{C \log d})$.

\textbf{Proof of $\lambda$ bound:} \Cref{thm:keyA} gives that the new
eigenvalues of $G_i$ are bounded by:
$2\sqrt{d - 1} + O \left(1 / (\log\log N)^4 \right ) + O(\sqrt{d} /
C^2)$. By picking $C = O(d^{1/4} / \sqrt{\epsilon})$ we get the
desired bound of $2\sqrt{d - 1} + \epsilon$.

Note that our choice of $C$ gives a final running time of
$\poly(N^{d^{1/4} \log d / \sqrt{\epsilon}})$, and this concludes the
proof of \Cref{thm:us}.

\subsection{The probabilistically strongly explicit construction}

We now finally prove \Cref{thm:prob-strong-explicit}, but first we point out a small caveat. Even though the theorem statement requires a construction where all but a $o(1)$ fraction of possible seeds result in a simple graph $G$ with the desired eigenvalue bound, the proof we present cannot achieve both simplicity and an $o(1)$ fraction of such good seeds. The reason for this lies in \Cref{sec:bordenave-derand}, since a graph sampled from the configuration model (or any $(\delta, k)$-wise uniform configuration model) is simple with constant probability, as seen in the proof of \Cref{thm:derand-bord}. Note, however, that the remaining two properties, namely bicycle-freeness at radius $r = c\log_{d-1}n$ and the eigenvalue bound, both hold with probability $1 - o(1)$.

Given this caveat, we offer two solutions to the reader: first, we note that if simplicity is not required then our construction holds with probability $1 - o(1)$. We do not prove that our construction still works in this scenario, the proof is straightforward but technical. Alternatively, we note that a different way of doing the first step of our construction is to do the following: derandomize a random $n$-lift of a $K_{d + 1}$, using Bordenave~\cite{Bor19}. It is a fact that a random $n$-lift of a $K_{d + 1}$ has all eigenvalues except the trivial one bounded by $2\sqrt{d - 1} + \epsilon$, for any $\epsilon > 0$, with high probability. The proof of this is very similar to the one we derandomize in \Cref{sec:bordenave-derand}, so for the sake of briefness we will not redo it. The benefit of using the lifts model is that lifting $K_{d + 1}$ never produces non-simple graphs, so the the rest of our construction still applies with the correct probability.

The rest of this section is devoted to prove \Cref{thm:us} but with the above caveat, that is, we shall ignore the simplicity constraint.

As before, assume we are given $N, d \geq 3, \epsilon > 0$. The structure of our construction is the same as the weakly explicit one, we simply make every step probabilistically strongly explicit, so recall the plan detailed right in the beginning of \Cref{sec:final-construction} and also the notation $G_i$ for the different iterations. As before, the number of vertices in the resulting graph is $\left \lceil 2^{(\log \log N)^4} \right \rceil \cdot \left \lceil 2^{\log N - (\log \log N)^4} \right \rceil$, so denote this quantity by $N'$. For simplicity we assume that the two previous factors are powers of 2, so, for instance, when we write $(\log\log N)^4$ we actually mean $\log \left \lceil 2^{(\log \log N)^4} \right \rceil$. Let's also assume we are given a seed $s \in \{0, 1\}^{O(\log N)}$, some $u \in [N']$ and $i \in [d]$. We will describe a $\poly(n)$ time algorithm (recall $N = 2^n$) that outputs the $i$th neighbor of vertex $u$ in $G(s)$, where $G(s)$ is $d$-regular graph that satisfies $\lambda(G(s)) \leq 2\sqrt{d - 1} + \epsilon$ for a $1 - o(1)$ fraction of the possible seeds $s$.

We start by considering $G_0$, which requires a $(m^{-8Cd^{1/4}\log(m)/\sqrt{\epsilon}}, C\log(m)d^{1/4}/\sqrt{\epsilon})$-wise uniform permutation $\bpi$ of length $dm$. As we saw before, this only requires a seed of length $o(\log N)$. Now, consider the binary form of the identifiers of each vertex in $G$. We interpret the first $(\log\log N)^4$ bits of a vertex in $G$ as the vertex in $G_0$ that gets lifted/copied to it. Let's denote by $u_0$ the vertex in $G_0$ that corresponds to $u$. Each of the remaining bits corresponds to one of the 2-lift iterations, namely which vertex copy does this vertex pertain to. We can use a similar interpretation scheme for the bits of each edge identifier.

In our construction, we set the $i$th neighbor of $u$ to be some copy of the $i$th neighbor of $u_0$ produced by the several 2-lifts we apply, so in particular we need to find this neighbor of $u_0$. Recall \Cref{def:config-model}, particularly how $\bpi$ is defined. Applying \Cref{thm:knr}, we can find $\bpi^{-1}(u_0, i)$ \pnote{We have to update \Cref{thm:knr} with the inverse permutation thing} in time
$\poly((\log(m)d^{1/4}/\epsilon) \log(m^{d^{1/4}\log(m)/\sqrt{\epsilon}})) = \poly(\log m) = o(\log N)$. Let $a = \bpi^{-1}(u_0, i)$. Now, if $a$ is odd then $\bpi(a + 1) = (v_0, j)$, where $v_0$ is the desired neighbor of $u_0$ (and $j$ indicates that $u_0$ is the $j$th neighbor of $u_0$). Similarly, if $a$ is even, $\bpi(a - 1) = (v_0, j)$. Note that $\bpi(a \pm 1)$ can also be found in time $o(\log N)$, according to \Cref{def:config-model}. Also, we have the same $\lambda$ guarantees as in \Cref{sec:final-construction}.

All that is left is to find the correct lifted copy of $v_0$, which we will describe inductively/recursively on the lift iterations given by \Cref{thm:keyA}. Let's denote by $u_j$ the vertex in $G_j$ that gets lifted to $u$ and define $v_0$ analogously. According to our interpretation of the binary form of $u_j$, this is just the first $(\log\log N)^4 + j$ first bits of $u$. So, suppose we know $u_j$ and $v_{j - 1}$, we shall see how to find $v_j$.

First, let $\kappa = d^{1/4} \log d / \sqrt{\epsilon}$, which was the polynomial exponent we obtained in the weakly explicit construction. We need a binary string $\by$ of length $d 2^{i - 2} m$ of $\left (N^{-O(\kappa)}, O(d^{1/4}/\sqrt{\epsilon}) \log \left (2^{i - 1}m \right) \right)$-wise uniform bits in order to apply \Cref{thm:keyA}. As we saw before, this requires a seed of length $O(\kappa \log N)$. Applying \Cref{thm:nn93}, we can find the $i$th bit of $\by$ in time $\polylog(N^{O(\kappa)}) = O(\kappa) \polylog(N)$, which we denote by $\by_i$. Suppose that the edge between $u_{j - 1}$ and $v_{j - 1}$ is the $a$th edge of $G_{j - 1}$. Notice how, since we have a similar interpretation scheme for the bits of edge identifiers, finding $a$ is analogous to finding $u_j$. Finally, given $\by_a$ it is trivial to find $v_j$: let $b$ be the last bit of $u_j$, if $\by_a$ is \texttt{0}, then $v_j = v_{j - 1} \cdot b$ otherwise $v_j = v_{j - 1} \cdot (\neg b)$, where $\cdot$ represents binary concatenation.

Our desired output vertex is $v_{\log N - (\log\log N)^4}$, the vertex we find after the last 2-lift iteration.

To finish our proof, we need to do some bookkeeping. Finding $v_0$ requires $o(\log N)$ time and a seed of length $o(\log N)$, so we can assign a $o(\log N)$ portion of our overall given seed $s$ to this step. Also, this step succeeds with probability at least $1 - m^{.99}$, according to \Cref{thm:derand-bord}. To find each $v_j$ with $j > 0$, we need time $O(\kappa) \polylog(N)$ and a seed of length $O(\kappa \log N)$. Each iteration succeeds with probability at least $1 - m^{-100}$, thus we can just reuse the same seed of length $O(\kappa \log N)$. We find the total probability of success by union bound: the total \textit{failure} probability is at most $\sum^{\log N} m^{-100} \leq m^{-100} \log N = o(1)$. Thus the total probability of success is $1 - o(1)$, or in other words, for a $1 - o(1)$ fraction of seeds $s$ this construction succeeds. Finally, note that the total running time is $O(\kappa) \polylog(N)$.

Since, we have the same $\lambda$ guarantees as in \Cref{sec:final-construction}, this concludes the proof.
}

\end{document}